\documentclass{article}
\usepackage[utf8]{inputenc}
\usepackage[margin=1in]{geometry}
\usepackage{rotating}
\usepackage{booktabs}
\usepackage{array}
\usepackage{amsmath, amsfonts, amssymb}
\usepackage[table,xcdraw]{xcolor}
\usepackage{hyperref}
\usepackage{bbm}
\usepackage{algorithm}
\usepackage[noend]{algpseudocode}
\usepackage{subcaption}
\usepackage{graphbox}
\usepackage{accents}
\usepackage{kpfonts}

\usepackage{tcolorbox}
\definecolor{c1}{cmyk}{0,0.6175,0.8848,0.1490} 
\definecolor{c2}{cmyk}{0.1127,0.6690,0,0.4431} 
\definecolor{c3}{cmyk}{0.3081,0,0.7209,0.3255} 
\definecolor{c4}{cmyk}{0.6765,0.2017,0,0.0667} 
\definecolor{c5}{cmyk}{0,0.8765,0.7099,0.3647} 
\definecolor{forestgreen}{HTML}{397727}

\newtcbox{\hlprimarytab}{on line, rounded corners, box align=base, colback=c3!10,colframe=white,size=fbox,arc=3pt, before upper=\strut, top=-2pt, bottom=-4pt, left=-2pt, right=-2pt, boxrule=0pt}
\newtcbox{\hlsecondarytab}{on line, box align=base, colback=red!10,colframe=white,size=fbox,arc=3pt, before upper=\strut, top=-2pt, bottom=-4pt, left=-2pt, right=-2pt, boxrule=0pt}

\newtcbox{\hlorangetab}{on line, box align=base, colback=orange!10,colframe=white,size=fbox,arc=3pt, before upper=\strut, top=-2pt, bottom=-4pt, left=-2pt, right=-2pt, boxrule=0pt}

\newtcbox{\hlgraytab}{on line, rounded corners, box align=base,colframe=white,size=fbox,arc=3pt, before upper=\strut, top=-2pt, bottom=-4pt, left=-2pt, right=-2pt, boxrule=0pt}

\usepackage{hhline}
\usepackage{colortbl}

\usepackage{tabularx,booktabs}
\newcolumntype{Y}{>{\centering\arraybackslash}X}

\usepackage{makecell}

\usepackage[numbers,sort]{natbib}

\usepackage{mathtools}
\mathtoolsset{showonlyrefs}

\usepackage{amsthm}
\usepackage[normalem]{ulem}
\usepackage{soul}

\newtheorem{theorem}{Theorem}

\newtheorem{assumption}{Assumption}

\usepackage{chngcntr}
\usepackage{apptools}
\AtAppendix{\counterwithin{theorem}{section}}
\AtAppendix{\counterwithin{prop}{section}}
\AtAppendix{\counterwithin{cor}{section}}
\AtAppendix{\counterwithin{lemma}{section}}

\usepackage{mathtools}

\usepackage{multicol}
\usepackage{multirow}

\definecolor{Gray}{gray}{0.8}
\definecolor{LightGray}{gray}{0.95}

\def\R{\mathbb{R}}
\def\E{\mathbb{E}}

\def\N{\mathbb{N}}
\def\cT{\mathcal{T}}

\def\cP{\mathcal{P}}
\def\cG{\mathcal{G}}
\def\ci{\mathrm{CI}}

\newcommand{\cN}{\mathcal{N}}

\newcommand{\CI}{\mathrm{CI}}

\renewcommand{\D}{\mathcal{D}}

\usepackage{esvect}

\makeatletter
\def\blfootnote{\xdef\@thefnmark{}\@footnotetext}
\makeatother

\title{Valid Inference with Synthetic Data\\ via Task Exchangeability}

\author{Lezhi Tan$^*$ \qquad Tijana Zrnic$^{\dagger,*}$\\ \\
$^*$Department of Management Science \& Engineering\\
$^\dagger$Department of Statistics\\ \\
Stanford University
}

\date{}

\begin{document}

\maketitle

\begin{abstract}
There is a proliferation of work arguing for the use of synthetic data in scientific research. For example, social scientists are arguing for the use of LLM-generated ``silicon samples'' in pilot studies; AI evaluations increasingly rely on ``LLM-as-a-judge'' outputs; and proteomics research is accelerated by generative models that produce synthetic protein structures. These developments raise an intriguing possibility: synthetic data may help researchers ask more questions, run more studies, and accelerate discovery. But they also raise a fundamental concern: synthetic data can be biased, noisy, and misspecified.
 In this work, we propose statistical principles for using synthetic data in scientific research with provable validity guarantees. The key insight is a new technical condition that we call \emph{task exchangeability}. Informally, this is a requirement that the researcher can identify historical tasks, for which real data is available, such that their current task of interest is exchangeable with the historical tasks in an appropriate mathematical sense. We develop methods for valid inference under task exchangeability, together with extensions that provide guarantees even beyond exchangeability. We demonstrate the framework on public opinion surveys with silicon samples and AI evaluation with autoraters.
\end{abstract}

\section{Introduction}

Across the sciences, researchers are increasingly turning to synthetic data: artificially generated data designed to mimic the statistical properties of real-world data. The appeal is easy to see. Synthetic data can help researchers study rare events, circumvent privacy constraints on sensitive data, and probe settings where real data is scarce or expensive. Social scientists are thus exploring LLM-generated ``silicon samples'' as a scalable way to study human behavior \cite{horton2023large, park2023generative}; AI evaluation increasingly relies on ``LLM-as-a-judge'' assessments \cite{zheng2023judging}; and generative models are accelerating biological discovery by producing synthetic protein structures \cite{jumper2021highly}.

Synthetic data, however, can be biased, noisy, and misaligned with real human judgment. For example, \citet{bisbee2024synthetic} show that LLM-generated synthetic ``respondents'' can match some aggregate patterns in survey data while failing to reproduce crucial inferential features, including variation across individuals and sensitivity to context. \citet{santurkar2023whose} show that the opinions reflected by LLMs can be substantially misaligned with those of many US demographic groups. More broadly, synthetic data can look plausible while still being wrong in the ways that matter for statistical inference. These concerns do not mean that synthetic data should be discarded altogether; they mean that standard statistical tasks---forming confidence intervals, testing hypotheses, and quantifying uncertainty---cannot be justified by simply pretending that synthetic data is real data.

Indeed, formally justifying inference with synthetic data is challenging. Suppose the scientific target is a property of a real-world distribution $\cP^*$, but the researcher only has access to synthetic samples $
\tilde X_1,\ldots,\tilde X_N \sim~\tilde \cP$.
Without assumptions relating $\cP^*$ and $\tilde \cP$, inference about $\cP^*$ from samples drawn from $\tilde \cP$ is, of course, impossible. Therefore, any solution allowing valid inference from synthetic data must rely on assuming or learning the relationship between $\cP^*$ and $\tilde \cP$.

Our starting point is a simple and intuitive solution to this problem. If real data is unavailable for the target task, then one can look for related historical tasks where real data was available. On these historical tasks, the researcher can compare real-data conclusions to synthetic-data conclusions, and learn how wrong the synthetic data tends to be. Then, for the target task, the researcher can begin with the naive confidence interval that treats the synthetic data as real, and enlarge it using the errors learned from the historical tasks.

We endow this intuitive solution with formal guarantees. In particular, we establish \emph{task exchangeability} as the key condition that makes this reasoning valid. 
Task exchangeability is an exchangeability condition in the formal statistical sense on the sequence of tasks and their corresponding datasets.
Under this condition, the empirical distribution of historical errors can be used to transform naive intervals based on synthetic data into confidence intervals with formal coverage guarantees for the true target quantity.

We develop methods for valid inference under task exchangeability, together with extensions that provide guarantees even beyond strict exchangeability. In particular, we show how the coverage guarantee degrades gracefully as the exchangeability assumption is violated. The resulting framework enables valid inference with synthetic data, no matter how misspecified this data may be.

\subsection{Preview of inference via task exchangeability}
\label{sec:preview}

We now give a brief preview of the main idea. Suppose we want to learn a target estimand $\theta^* = \theta^*(\cP^*)$. We have no data from the corresponding true distribution $\cP^*$, but we do have a synthetic dataset $\tilde S = \{\tilde X_1,\dots,\tilde X_N\}$, drawn from some synthetic distribution $\tilde{\mathcal P}$. From this synthetic sample, we can estimate the corresponding synthetic-data estimand $\tilde \theta = \theta^*(\tilde \cP)$. Using any standard method such as the central limit theorem or Hoeffding's inequality, we can construct a confidence interval $\widetilde{\mathrm{CI}} = [\tilde L,\tilde U]$ such that
\[
P\left(
\tilde \theta \in \widetilde{\mathrm{CI}}
\right)
\geq
1-\alpha.
\]
Now suppose that we knew that $|\theta^* - \tilde \theta| \leq \Delta$, for some gap $\Delta$. Then, we could simply expand the synthetic-data interval by $\Delta$ in both directions:
\[
\mathrm{CI}
=
\widetilde{\mathrm{CI}} + [-\Delta,\Delta]
=
[\tilde L-\Delta,\tilde U+\Delta].
\]
Indeed,
\begin{align}
P\left(
\theta^* \in \mathrm{CI}
\right)
&\geq
 P\left(
\tilde \theta \in \widetilde{\mathrm{CI}},
|\theta^*-\tilde \theta|\leq \Delta
\right)
\geq
1-\alpha.
\end{align}
Thus, if we knew how far the synthetic estimand could be from the real estimand, valid inference would be immediate: compute the naive synthetic-data interval, and inflate it by this gap.

The challenge, of course, is that $\Delta$ is not known. Task exchangeability provides a way to infer this gap from related historical tasks. Suppose we have historical tasks $j=1,\dots,T$, for which both real and synthetic data are available, allowing us to measure the corresponding absolute gaps between the synthetic and real estimands $
\Delta_j$.
The target gap is $
\Delta_{T+1} \equiv \Delta$.
Under task exchangeability, formally defined in the next section, the target discrepancy $\Delta_{T+1}$ is exchangeable with the historical gaps, meaning $\Delta_1,\ldots,\Delta_{T+1}$ is an exchangeable sequence of random variables. Consequently, up to finite-sample corrections, $\Delta_{T+1}$ is at most the $(1-\beta)$-quantile of the historical discrepancies with probability at least $1-\beta$. Denoting this empirical quantile by
\[
\hat \Delta
=
Q_{1-\beta}\left(
\Delta_1,\ldots,\Delta_T
\right),
\]
we can form the expanded interval
\[
\mathrm{CI}
=
[\tilde L-\hat \Delta,\tilde U+\hat \Delta].
\]
Combining the validity of the synthetic-data interval with the validity of the upper bound on the gap $\Delta$ yields
\[
P(
\theta^* \in \mathrm{CI}
)
\geq P(
\tilde \theta \in \widetilde{\mathrm{CI}}, \Delta \leq \hat \Delta
) \geq
1-\alpha-\beta,
\]
by a union bound.

To sum up, we have shown that we can infer $\theta^*$ by relying on real data from the historical tasks, without requiring any real data from the target task.
Our main proposed method follows this general template, with additional finite-sample corrections making the argument fully rigorous.

\begin{figure}[t]
\centering
\includegraphics[width=0.9\textwidth]{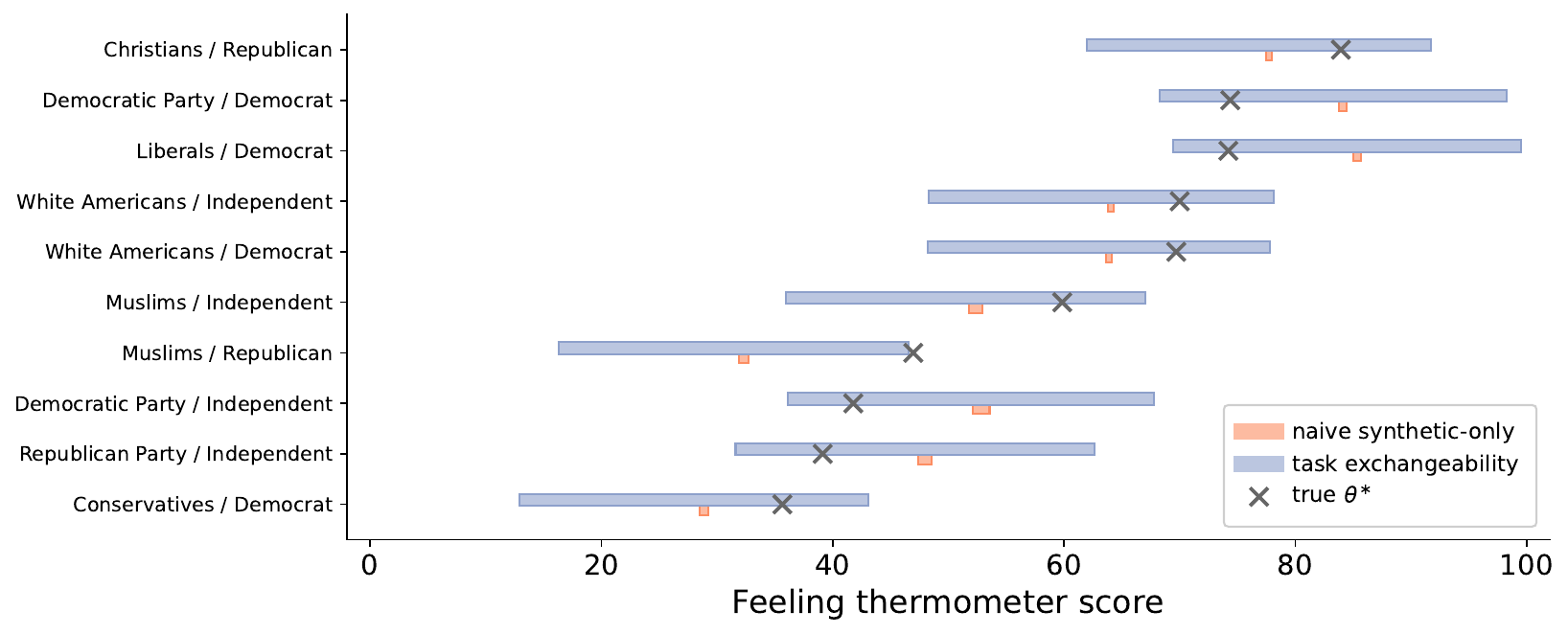}
\caption{\textbf{Inference on ANES feeling-thermometer scores via task exchangeability (preview).}
Each row corresponds to a task defined by a target group and respondent subgroup, with the estimand equal to the average ANES feeling-thermometer score on a 0--100 scale. The naive synthetic-only intervals treat the synthetic data as real data; our work proposes the task-exchangeability intervals.}
\label{fig:anes_intervals}
\end{figure}

As an empirical preview, Figure~\ref{fig:anes_intervals} illustrates inference via task exchangeability using data from \citet{bisbee2024synthetic}, who study whether large language models can be used to generate synthetic survey responses. The underlying real data comes from the American National Election Studies (ANES), a long-running survey program measuring political attitudes and behavior in the United States. We focus on ANES feeling-thermometer questions, in which respondents rate groups or political parties on a scale from 0 to 100, with higher values indicating warmer feelings. Each task corresponds to estimating the average thermometer score for a particular target group within a respondent subgroup, such as views toward Republicans among Republican respondents or views toward Muslims among Independents.

For each task, we compare the naive synthetic-only interval, which treats the synthetic data as real data, with the interval obtained via task exchangeability. We set the target overall coverage level to be $0.85$. We treat survey results from the previous survey wave, reporting analogous feeling thermometer scores, as the historical tasks exchangeable with the one of interest. The synthetic-only intervals are too narrow and severely biased, missing the true estimand $\theta^*$. This is consistent with the findings of Bisbee et al., who observe that language model responses exhibit less variation than real surveys.
By contrast, the task-exchangeability intervals are wider, reflecting uncertainty about the real--synthetic gap, and cover the true values at the desired rate.

\subsection{Related work}

We discuss several methodological developments most related to our proposal.

\paragraph{Prediction-powered inference.}
Our work is conceptually related to prediction-powered inference (PPI) \cite{angelopoulos2023prediction, angelopoulos2023ppipp, zrnic2023cross, fisch2024stratified, cowen2026multiple, song2026demystifying, chen2026power}, which provides valid inference by combining gold-standard data with abundant machine learning predictions. In its standard form, PPI uses a small labeled sample from the target population to debias an estimator computed using machine-learning-predicted labels on a larger unlabeled sample. Several recent works have studied variants of PPI with imputed covariates, rather than imputed labels \citep{kluger2025prediction, miao2025assumption, ji2025predictions, zhao2025imputation}. However, they still require a complete gold-standard sample, as well as a subset of observed covariates everywhere that are used to compute the predictions. In our setting, the data for the target task is fully synthetic, and the real--synthetic discrepancy is calibrated using historical tasks rather than real data from the same task.
Notably, a few recent works have studied settings with many related tasks. \citet{emmenegger2026prediction} study prediction-powered inference across tasks, showing how information from related tasks can be used to recalibrate predictions while preserving task-specific validity. \citet{li2025prediction} combine prediction-powered inference with empirical Bayes ideas, using shrinkage across related problems to improve estimation.
These works share with ours the idea that related tasks can provide useful information about the reliability of synthetic measurements. In the PPI literature, the synthetic quantity is typically an imputed label or covariate attached to other real covariates, while in our setting the data points are fully synthetic.

\paragraph{Inference with synthetic data.}
Several recent papers study valid inference with imperfect synthetic data. \citet{byun2026valid} develop a method-of-moments approach for inference using synthetic samples, using real data to correct for the bias. The synthetic data is treated as an auxiliary variable to reduce the variance of the real-data estimator, not as entirely fresh samples.
\citet{bashari2025statistical} propose GESPI, a general framework for distribution-free inference with synthetic data that relies on a guardrail that protects against relying too heavily on poor synthetic data. Our work belongs to the same broader agenda of developing inference procedures that can benefit from synthetic data without treating it as if it were real. A distinctive feature of our work is that we do not assume access to any real data from the task of interest; the idea is to handle synthetic datasets generated from scratch. We therefore use historical tasks to calibrate this gap, with task exchangeability serving as the condition that makes this calibration valid. Our framework can be extended to settings where some target-task real data is available; we discuss this case and the connection to \cite{bashari2025statistical} in Section~\ref{sec:some_real_data}.

Most closely related to our setting, \citet{huang2025many} study inference about human survey responses when real responses are available for a collection of calibration questions but not for the target question. Their method selects the number of LLM-simulated responses so that confidence sets constructed from synthetic data contain real-data confidence sets sufficiently often across the calibration questions. By limiting the synthetic sample size, the method retains enough simulation noise for the resulting confidence set to account for human--LLM misalignment. Under an i.i.d.\ sampling model for the calibration and target questions, they establish coverage for a new question and interpret the selected sample size as the effective number of human respondents represented by the LLM.
Although the motivation is similar to ours, the resulting approaches are quite different. \citet{huang2025many} calibrate a synthetic sample size so that the bias of the synthetic data is dominated by its stochastic error. Their confidence intervals therefore use synthetic sampling variability as a buffer against human--LLM misalignment and remain centered at the synthetic estimator. We instead estimate the real--synthetic gap via exchangeability directly. Consequently, our intervals need not be centered at the synthetic estimator: the estimated gap can shift the interval to account for systematic bias. This also leads to the opposite prescription for the synthetic sample size. Rather than deliberately retaining synthetic sampling noise, our method favors using as much synthetic data as possible, since doing so sharpens estimation of both the synthetic estimand and the real--synthetic gaps observed on the historical tasks.

Other, more distant approaches to inference with synthetic data include those of \citet{decruyenaere2024real, keret2025glm, mccaw2024synthetic}. These papers generally rely on access to some real-data information from the task of interest, parametric modeling assumptions on the data-generating process, or specific inferential targets.

\paragraph{Conformal prediction and permutation testing.} Our approach is also connected to conformal prediction \cite{vovk2005algorithmic, angelopoulos2024theoretical} and permutation-based inference, where exchangeability is used to turn empirical ranks into finite-sample guarantees. The closest connection is by \citet{andrews2025transfer}, who study the transfer performance of economic models across domains. In their setting, domains are modeled as exchangeable, and the empirical distribution of transfer errors across observed domains is used to form prediction intervals for the transfer error on a new domain. We use a similar exchangeability principle at the level of tasks, but apply it to statistical inference with synthetic data for a new task.
We also draw on ideas from weighted variants of conformal prediction \cite{tibshirani2019conformal, fannjiang2022conformal, guan2023localized}. In particular, our weighted procedure can be viewed as a task-level analogue of weighted split conformal prediction by \citet{barber2023conformal}: historical tasks receive relevance weights, and departures from exchangeability appear as explicit total-variation penalties in the coverage guarantee.

\paragraph{Empirical Bayes.} 
Our framework is related in spirit to empirical Bayes methods \cite{robbins1964empirical, robbins1992empirical}. Empirical Bayes leverages common population-level structure across problems to improve task-specific estimation. Our work adopts a similar meta-distribution perspective: the current task and the historical tasks are viewed as related draws from a larger population of tasks. Rather than using this structure to estimate a prior or construct shrinkage estimators, we use it to calibrate the distribution of real--synthetic estimation gaps. In this sense, task exchangeability can be viewed as an empirical-Bayes-style source of information, but used for a different purpose.

\section{Preliminaries}

We describe the mathematical model used throughout the paper and formally define task exchangeability.

Let $\theta(\cP) \in \Theta$ denote a statistical target, defined by applying a functional $\theta(\cdot)$ to a data distribution $\cP$. The set $\Theta$ is a space of possible quantities of interest, such as means, medians, regression coefficients, and other population-level summaries. We refer to a functional--distribution pair $\cT = (\theta,\cP)$ as a \emph{task}. The current task of interest is denoted by $\cT^* = (\theta^*,\cP^*)$,
with target estimand $\theta^*(\cP^*)$. Slightly abusing notation, we will sometimes keep the distribution implicit and simply write $\theta^* \equiv \theta^*(\mathcal{P}^*)$.

Traditionally, to infer $\theta^*(\cP^*)$, one would collect a real dataset $
S = \{X_1,\dots,X_n\} \stackrel{\mathrm{i.i.d.}}{\sim} \cP^*$.
The setting of interest in this paper is one where collecting such a dataset is costly or challenging. As a potential remedy, we have access to a general-purpose synthetic data generator $\cG$---for example, a large language model. Given a task $\cT^* = (\theta^*,\cP^*)$, we can instantiate the generator for this task and draw synthetic samples from the resulting distribution. We denote this synthetic distribution by
$\tilde \cP = \cG_{\cT^*}.$
For example, $\tilde \cP$ could be the distribution of ``silicon samples,'' obtained by prompting a language model to mimic responses from a target demographic population. We then observe a synthetic dataset $\tilde S = \{\tilde X_1,\dots,\tilde X_N\}
\stackrel{\mathrm{i.i.d.}}{\sim}
\tilde \cP.$
The corresponding synthetic-data target is $
\tilde \theta = \theta^*(\tilde \cP).$

For a functional $\theta$ and dataset $S$, we write $\theta(S)$ to mean $\theta(\cdot)$ applied to the empirical distribution induced by $S$; for instance, if $\theta$ is a mean functional, then $\theta(S)$ is the sample mean. This is not strictly necessary, in the sense that other estimators could be applied to dataset $S$, but we stick to this choice for simplicity of exposition.

The key assumption we leverage is that we can identify a set of historical tasks $\cT_1, \dots, \cT_T$, for which we have real data, that are ``similar'' to $\cT^*$ in a precise mathematical sense. Each historical task is a tuple $\cT_j = (\theta_j,\cP_j)$.
For historical task $j$, we observe
$
S_j = \{X_1^j,\dots,X_{n_j}^j\}
\stackrel{\mathrm{i.i.d.}}{\sim}
\cP_j.$
We can also instantiate the same synthetic data generator for each historical task, obtaining $\tilde \cP_j := \cG_{\cT_j}$, from which we can draw synthetic samples $
\tilde S_j = \{\tilde X_1^j,\dots,\tilde X_N^j\}
\stackrel{\mathrm{i.i.d.}}{\sim}
\tilde \cP_j.$
We write $\theta_j = \theta_j(\cP_j)$ for the real-data target and $\tilde \theta_j = \theta_j(\tilde \cP_j)$ for the corresponding synthetic-data target.

Our core assumption is \emph{exchangeability} of the task--dataset pairs. Below, $S$ denotes the counterfactual real dataset for the current task $\cT^*$; this is the dataset we would have liked to collect from $\cP^*$, but do not observe.

\begin{assumption}[Task exchangeability]
\label{ass:task_exchangeability}
We assume that $(\cT_1, S_1),\dots, (\cT_{T+1}, S_{T+1})$, where $\cT_{T+1}=\cT^*, S_{T+1} = S$, are exchangeable.
\end{assumption}

In other words, if we shuffle the order of the task-dataset pairs, the distribution of the sequence remains the same. Notice that Assumption \ref{ass:task_exchangeability} assumes that tasks and datasets are intrinsically random and drawn from a ``meta-distribution.'' For example, the dataset sizes $n_j$ and the estimands $\theta_j$ are both random and drawn from a common underlying distribution.
It is also possible to assume that only $\cT_1,\dots,\cT_{T+1}$ are exchangeable, but the resulting procedure is more conservative. We discuss this extension in Appendix \ref{sec:weaker_condition}.

In Figure \ref{fig:anes_intervals}, there are $T=33$ historical tasks. Each $\cP_j$ is defined as the distribution of feeling thermometer scores for a particular target group within a respondent group; $\theta_j$ is simply the mean for all $j$. The synthetic generator $\cG$ is a GPT model, and it is instantiated for each task $\cT_j$ by prompting the model to mimic a respondent from the respondent group that defines $\cT_j$.

\section{Valid inference via task exchangeability}

We now describe our main procedure for valid inference on the current target $\theta^*=\theta^*(\cP^*)$ using synthetic data and historical exchangeable tasks. Throughout this section, we index the current task by $T+1$, so that $\cT_{T+1}=\cT^*$, $\theta_{T+1}=\theta^*$, and $\cP_{T+1}=\cP^*$. We use $S_{T+1}=S$ to denote the counterfactual real dataset for the current task, which is not observed.

The procedure has two ingredients. First, we use the synthetic data for the current task to estimate the synthetic-data target $\tilde\theta=\theta^*(\tilde\cP)$. Second, we use the historical tasks to estimate how far the real-data target tends to be from the corresponding synthetic-data target. For task $j$, define the real--synthetic gap
$\Delta_j = \theta_j(\cP_j)-\theta_j(\tilde\cP_j)$.
Our goal is to infer $\theta^*=\tilde\theta+\Delta_{T+1}$, even though the current real-data gap $\Delta_{T+1}$ cannot be directly estimated from real data.

We assume access to standard inference procedures for population targets.

\begin{assumption}[Classical inference methods]
\label{ass:classical_methods}
Fix any error level $\alpha\in(0,1)$, sample sizes $n,n'\in\N$, and functional $\theta\in\Theta$. Let $S=\{X_1,\dots,X_n\}\stackrel{\mathrm{i.i.d.}}{\sim}\cP$ and $S'=\{X_1',\dots,X_{n'}'\}\stackrel{\mathrm{i.i.d.}}{\sim}\cP'$ be independent datasets. We assume access to two valid confidence interval constructions:
\begin{itemize}
\item $\mathrm{CI}^{\theta,\alpha}(S)$, which satisfies $P\left(\theta(\cP)\in \mathrm{CI}^{\theta,\alpha}(S) \right)\geq 1-\alpha$;
\item $\Delta^{\theta,\alpha}(S,S')$, which satisfies $P\left(\theta(\cP)-\theta(\cP')\in \Delta^{\theta,\alpha}(S,S') \right)\geq 1-\alpha$.
\end{itemize}
\end{assumption}

The first procedure gives a confidence interval for a target under a single distribution. The second gives a confidence interval for the difference between the same functional evaluated under two distributions. In our setting, it will be applied with one real sample and one synthetic sample. To satisfy Assumption \ref{ass:classical_methods}, one can use generic confidence interval constructions such as the bootstrap or intervals based on the central limit theorem.

We state the main algorithm in Algorithm \ref{alg:main_algo}. We use $v_{(k)}$ to denote the $k$-th order statistic of $v_1,\dots,v_T$, with the conventions $v_{(0)}=-\infty$ and $v_{(T+1)}=\infty$. The algorithm follows the sketch given in Section \ref{sec:preview}, with a few additional finite-sample corrections.

\begin{algorithm}[t]
\caption{Inference via task exchangeability}
\label{alg:main_algo}
\begin{algorithmic}[1]
\Require current task $\cT^*$, historical task--dataset pairs $(\cT_1,S_1),\dots,(\cT_T,S_T)$, synthetic data generator $\cG$, synthetic sample size $N$, error levels $\alpha_1,\alpha_2,\alpha_3\in(0,1)$
\State Draw $N$ synthetic samples for the current task: $\tilde S\stackrel{\mathrm{i.i.d.}}{\sim}\tilde\cP$, where $\tilde\cP=\cG_{\cT^*}$
\State Draw $N$ synthetic samples for each historical task: $\tilde S_j\stackrel{\mathrm{i.i.d.}}{\sim}\tilde\cP_j$, where $\tilde\cP_j=\cG_{\cT_j}$, for all $j\in[T]$
\State Compute a confidence interval for the current synthetic-data target $\tilde \theta$:
$[\tilde L,\tilde U]=\mathrm{CI}^{\theta^*,\alpha_1}(\tilde S)$
\State For each historical task $j\in[T]$, compute a confidence interval for the gap $\Delta_j$:
$
[\hat\Delta_j^L,\hat\Delta_j^U]=\Delta^{\theta_j,\alpha_2}(S_j,\tilde S_j)$
\State Let $\hat\Delta^L=\hat\Delta_{(k_L)}^L$ and $\hat\Delta^U= \hat\Delta_{(k_U)}^U$, where $k_L=\lfloor (T+1)\frac{\alpha_3}{2}\rfloor$, and $k_U=\lceil (T+1)(1-\frac{\alpha_3}{2})\rceil$
\Ensure Confidence interval for $\theta^*$:
$\ci=[\tilde L+\hat\Delta^L,\tilde U+\hat\Delta^U]$
\end{algorithmic}
\end{algorithm}

\begin{theorem}
\label{thm:main_thm}
Suppose Assumptions~\ref{ass:task_exchangeability} and~\ref{ass:classical_methods} hold. Then, the confidence interval $\ci$ output by Algorithm~\ref{alg:main_algo} satisfies
\[
P(\theta^*\in\ci)\geq 1-(\alpha_1+\alpha_2+\alpha_3).
\]
\end{theorem}

\begin{proof}
First, by the validity of $\mathrm{CI}^{\theta^*,\alpha_1}$ in Assumption~\ref{ass:classical_methods}, we have 
\begin{equation}
\label{eq:alpha_1}
P(\tilde\theta\in[\tilde L,\tilde U])\geq 1-\alpha_1.
\end{equation}

Next, we show that $[\hat\Delta^L,\hat\Delta^U]$ covers the current gap $\Delta_{T+1}=\theta^*-\tilde\theta$ with probability at least $1-\alpha_2-\alpha_3$. Given the counterfactual dataset $S_{T+1}=S$ and $\tilde S_{T+1}=\tilde S$, let
\[
[\hat\Delta_{T+1}^L,\hat\Delta_{T+1}^U]
=
\Delta^{\theta^*,\alpha_2}(S,\tilde S).
\]
By Assumption~\ref{ass:classical_methods}, this interval satisfies
\begin{equation}
\label{eq:alpha_2}
P(\Delta_{T+1}\in[\hat\Delta_{T+1}^L,\hat\Delta_{T+1}^U])\geq 1-\alpha_2.
\end{equation}

We now compare the unobserved current gap interval $[\hat\Delta_{T+1}^L,\hat\Delta_{T+1}^U]$ to the observed historical gap intervals $[\hat\Delta_{j}^L,\hat\Delta_{j}^U]$. By task exchangeability, $(\cT_1,S_1),\dots,(\cT_{T+1},S_{T+1})$ are exchangeable. Since each synthetic dataset $\tilde S_j$ is drawn from a distribution determined by $\cT_j$, independently according to the same generator $\cG$ conditional on the tasks, the triples $(\cT_j,S_j,\tilde S_j)$ are also exchangeable. Finally, because $[\hat\Delta_j^L,\hat\Delta_j^U]$ is a deterministic function of $(\cT_j,S_j,\tilde S_j)$, the sequence of intervals $([\hat\Delta_j^L,\hat\Delta_j^U])_{j=1}^{T+1}$ is exchangeable.

Recall a fundamental consequence of exchangeability: given an exchangeable sequence $v_1,\dots,v_{T+1}$, $v_{T+1}$ has a uniform rank among $v_1,\dots,v_{T+1}$. This implies that, for any $\alpha\in(0,1)$, the probability that $\hat\Delta_{T+1}^L$ is among the smallest $\alpha$-fraction of $\hat\Delta_j^L$ is at most $\alpha$. More formally:
\[
P\left(\hat\Delta_{T+1}^L < \hat\Delta^L\right) = P\left(\hat\Delta_{T+1}^L < \hat\Delta^L_{(k_L)} \right)\leq \frac{\alpha_3}{2},
\]
where $k_L=\lfloor (T+1)\frac{\alpha_3}{2}\rfloor$. Similarly, since $\hat\Delta^U = \hat\Delta^U_{(k_U)}$, with $k_U=\lceil (T+1)(1-\frac{\alpha_3}{2})\rceil$, exchangeability gives
\[
P\left(\hat\Delta_{T+1}^U > \hat\Delta^U \right)\leq \frac{\alpha_3}{2}.
\]
Combining the two displays by a union bound, we obtain \begin{equation}
\label{eq:alpha_3}
P\left([\hat\Delta_{T+1}^L,\hat\Delta_{T+1}^U]\subseteq[\hat\Delta^L,\hat\Delta^U]\right)\geq 1-\alpha_3.
\end{equation}
Therefore, combining equations~\eqref{eq:alpha_2} and \eqref{eq:alpha_3} through another union bound gives
\begin{equation}
\label{eq:alpha_23}
P\left(\Delta_{T+1}\in[\hat\Delta^L,\hat\Delta^U]\right)
\geq
1-\alpha_2-\alpha_3.
\end{equation}
On the event that $\tilde\theta\in[\tilde L,\tilde U]$ and $\Delta_{T+1}\in[\hat\Delta^L,\hat\Delta^U]$, we have \[\theta^*=\tilde\theta+\Delta_{T+1}\in[\tilde L+\hat\Delta^L,\tilde U+\hat\Delta^U]=\ci.\]
Therefore, combining equations \eqref{eq:alpha_1} and \eqref{eq:alpha_23} yields
\[
P(\theta^*\in\ci) \geq P\left(\tilde\theta\in[\tilde L,\tilde U],\  \Delta_{T+1}\in[\hat\Delta^L,\hat\Delta^U]\right) \geq 1-(\alpha_1+\alpha_2+\alpha_3).\]
\end{proof}

As a rule of thumb, we can take $\alpha_1$ to be a tiny fraction of the error budget, since typically we can generate a large, virtually unlimited number of data points using the generator $\cG$. Therefore, the error budget should mostly be split between $\alpha_2$ and $\alpha_3$.

A similar exchangeability argument can be used to infer the finite-sample target $\theta^*(S)$, rather than the full population target $\theta^*(\cP^*)$. This algorithm is, in fact, simpler than Algorithm \ref{alg:main_algo}, and does not require splitting the error budget. Moreover, its coverage is within $[1-\alpha, 1-\alpha + \frac{2}{T+1}]$ under a mild assumption, meaning the method does not conservatively overcover. We include the details in Section \ref{sec:finite_sample_target}.

\paragraph{A better point estimate.}
The task-exchangeability framework also suggests a natural point estimate of the target estimand $\theta^*$.
For each historical task $j \in [T]$, let
\[
\hat \Delta_j =  \theta_j(S_j) - \theta_j(\tilde S_j)
\]
denote an estimate of the bias $\Delta_j = \theta_j - \tilde \theta_j$, computed from the real and synthetic data available for that task. We then estimate the target-task bias by the historical average
\[
\bar \Delta = \frac{1}{T}\sum_{j=1}^T \hat \Delta_j,
\]
and define the bias-corrected estimator
\[
\hat \theta = \theta^*(\tilde S) + \bar \Delta.
\]

This estimator is motivated by the same exchangeability assumption as the interval procedure. Since the tasks are drawn from the same meta-distribution, we have $\E[\theta_j(\tilde S_j)] = \E[\theta^*(\tilde S)]$, and thus $\hat\theta$ is unbiased for the finite-sample target $\theta^*(S)$, in the sense that $\E[\hat\theta - \theta^*(S)] = 0$.

\section{Extensions}

We discuss several extensions of our main approach from the previous section. In particular, we consider inference beyond exchangeability, multi-dimensional inference targets, inference on the finite-sample target $\theta^*(S)$, and how to adapt the approach when some amount of real data from the target task is available.

\subsection{Valid inference beyond exchangeability}
\label{sec:beyond_exchangeability}

The guarantee in Theorem~\ref{thm:main_thm} relies on exchangeability between the historical tasks and the current task. Often, this assumption may be too strong: the historical tasks may be informative about the current task without being exactly exchangeable with it. We now show that the same calibration idea yields a valid, though weakened, guarantee under approximate exchangeability.

In doing so, we study a more general procedure. Algorithm~\ref{alg:main_algo} calibrates the current real--synthetic gap using the empirical distribution of the historical gaps. We now describe a weighted version of this calibration step. The weights allow the user to specify which historical tasks are most relevant to the current task. The unweighted procedure from Algorithm~\ref{alg:main_algo} is recovered as a special case by assigning equal weights to all historical tasks.

Let $w_1,\dots,w_T\in[0,1]$ be fixed weights assigned to the historical tasks, and define normalized weights
\[
\bar w_i=\frac{w_i}{1+\sum_{j=1}^T w_j},
\quad i\in[T],
\qquad
\bar w_{T+1}=\frac{1}{1+\sum_{j=1}^T w_j}.
\]
The condition $w_i\leq 1$ ensures $\bar w_i\leq \bar w_{T+1}$, so the current task receives at least as much weight as any individual historical task.

We consider setting $\hat \Delta^L$ and $\hat \Delta^U$ as \emph{weighted} quantiles of the historical gaps. For a probability distribution $\nu$ on the real line, let $Q_\gamma(\nu)$ denote its $\gamma$-quantile.\footnote{When the $\gamma$-quantile is not unique, we use the largest possible $\gamma$-quantile for the lower endpoints and the smallest possible $\gamma$-quantile for the upper endpoints.} We define the weighted endpoints at error level $\gamma$ as
\[
\hat\Delta^L
=
Q_{\gamma}
\left(
\sum_{j=1}^T \bar w_j \delta_{\hat\Delta_j^L}
+
\bar w_{T+1}\delta_{-\infty}
\right), \qquad \hat\Delta^U
=
Q_{1-\gamma}
\left(
\sum_{j=1}^T \bar w_j \delta_{\hat\Delta_j^U}
+
\bar w_{T+1}\delta_{+\infty}
\right).
\]
The atoms at $+\infty$ and $-\infty$ play the same role as the finite-sample order-statistic conventions in Algorithm~\ref{alg:main_algo}; they account for the fact that the current-task endpoint is unobserved. We state the full procedure in Algorithm \ref{alg:weighted_exchangeability}. The only difference compared to Algorithm~\ref{alg:main_algo} is the weighted computation of $\hat\Delta^L$ and $\hat\Delta^U$.

Let
\[
V^L = (\hat\Delta_1^L,\dots,\hat\Delta_T^L,\hat\Delta_{T+1}^L),
\qquad
V^U = (\hat\Delta_1^U,\dots,\hat\Delta_T^U,\hat\Delta_{T+1}^U)
\]
denote the vectors of lower and upper endpoints of the confidence intervals for the gaps $\Delta_j$, including the counterfactual endpoints for the current task. For each $i\in[T]$ let $V^{U,i}$ denote the vector obtained by swapping the $i$-th and $(T+1)$-st entries of $V^U$, and define $V^{L,i}$ analogously. Let
\[
\varepsilon_L
=
\sum_{i=1}^T
\bar w_i
d_{\mathrm{TV}}\left(V^L,V^{L,i}\right), \qquad
\varepsilon_U
=
\sum_{i=1}^T
\bar w_i
d_{\mathrm{TV}}\left(V^U,V^{U,i}\right).
\]
Here, $d_{\mathrm{TV}}$ denotes the total variation distance between the distributions of its arguments.
The quantities $\varepsilon_L$ and $\varepsilon_U$ measure the extent to which the endpoint vectors fail to be invariant to swapping the current task with a weighted historical task. Under exact exchangeability, both $\varepsilon_L$ and $\varepsilon_U$ are zero.

\begin{algorithm}[t]
\caption{Weighted inference beyond task exchangeability}
\label{alg:weighted_exchangeability}
\begin{algorithmic}[1]
\Require current task $\cT^*$, historical task--dataset pairs $(\cT_1,S_1),\dots,(\cT_T,S_T)$, synthetic data generator $\cG$, synthetic sample size $N$, error levels $\alpha_1,\alpha_2,\alpha_3\in(0,1)$, weights $w_1,\dots,w_T\in[0,1]$
\State Draw $N$ synthetic samples for the current task: $\tilde S\stackrel{\mathrm{i.i.d.}}{\sim}\tilde\cP$, where $\tilde\cP=\cG_{\cT^*}$
\State Draw $N$ synthetic samples for each historical task: $\tilde S_j\stackrel{\mathrm{i.i.d.}}{\sim}\tilde\cP_j$, where $\tilde\cP_j=\cG_{\cT_j}$, for all $j\in[T]$
\State Compute a confidence interval for the current synthetic-data target $\tilde\theta$:
$[\tilde L,\tilde U]=\mathrm{CI}^{\theta^*,\alpha_1}(\tilde S)$
\State For each historical task $j\in[T]$, compute a confidence interval for the gap $\Delta_j$:
$
[\hat\Delta_j^L,\hat\Delta_j^U]=\Delta^{\theta_j,\alpha_2}(S_j,\tilde S_j)
$
\State Define normalized weights
$
\bar w_i=\frac{w_i}{1+\sum_{j=1}^T w_j}
$
for $i\in[T]$, and
$
\bar w_{T+1}=\frac{1}{1+\sum_{j=1}^T w_j}
$
\State Let
$\hat\Delta^L
=
Q_{\frac{\alpha_3}{2}}
\left(
\sum_{j=1}^T \bar w_j\delta_{\hat\Delta_j^L}
+
\bar w_{T+1}\delta_{-\infty}
\right)$ and $\hat\Delta^U
=
Q_{1-\frac{\alpha_3}{2}}
\left(
\sum_{j=1}^T \bar w_j\delta_{\hat\Delta_j^U}
+
\bar w_{T+1}\delta_{+\infty}
\right)$
\Ensure Confidence interval for $\theta^*$:
$\ci=[\tilde L+\hat\Delta^L,\tilde U+\hat\Delta^U]$
\end{algorithmic}
\end{algorithm}

\begin{theorem}
\label{thm:beyond_exchangeability}
Suppose Assumption~\ref{ass:classical_methods} holds. Then, the confidence interval $\ci$ output by Algorithm~\ref{alg:weighted_exchangeability} satisfies
\[
P(\theta^*\in\ci)
\geq
1-(\alpha_1+\alpha_2+\alpha_3)-(\varepsilon_L+\varepsilon_U).
\]
\end{theorem}

Theorem~\ref{thm:beyond_exchangeability} recovers the guarantee from Theorem \ref{thm:main_thm} under task exchangeability, in which case $\varepsilon_L=\varepsilon_U=0$. More generally, the theorem gives a robustness statement: violations of exchangeability degrade coverage through the weighted total-variation penalties. Thus, choosing weights that concentrate on historical tasks closer to the current task can tighten the guarantee. Theorem~\ref{thm:beyond_exchangeability} is inspired by an analogous theorem bounding the coverage of weighted conformal prediction beyond exchangeability \cite{barber2023conformal}.

\subsection{Multi-dimensional targets}
\label{sec:multi_dimensional_targets}

Algorithm~\ref{alg:main_algo} was stated for scalar targets. We now describe the corresponding construction when the target is vector-valued. Suppose that $\Theta\subseteq\R^d$, so that
\[
\theta_j(\cP_j),\; \theta_j(\tilde\cP_j),\; \Delta_j
\in\R^d.
\]
We use the same notation $\mathrm{CI}^{\theta,\alpha}(S)$ and $\Delta^{\theta,\alpha}(S,S')$ as in Assumption~\ref{ass:classical_methods}, but interpret them as set-valued confidence regions in $\R^d$ rather than intervals. Thus, for example, $\Delta^{\theta,\alpha}(S,S')$ is a confidence region for the vector $\theta(\cP)-\theta(\cP')$. We assume these regions are convex and compact.

The main difference from the scalar case is that there is no canonical ordering of vectors, and hence no canonical lower and upper quantile for $\Delta_j$. We therefore avoid defining a multivariate quantile directly. Instead, we choose a nested family of convex regions and reduce each historical gap confidence region to one scalar score.

Let $\{\D_r:r\in\R\}$ be a fixed family of subsets of $\R^d$ satisfying the following conditions: for each finite $r$, $\D_r$ is convex and compact; if $r\leq r'$, then $\D_r\subseteq\D_{r'}$; $\bigcup_{r\in\R}\D_r=\R^d$; and for every compact set $K\subseteq\R^d$, the value
\[
R(K):=\inf\{r:K\subseteq\D_r\}
\]
satisfies $K\subseteq\D_{R(K)}$. The last condition is a mild right-continuity condition on the family. We also use the convention $\D_{+\infty}=\R^d$.
For a point $\delta\in\R^d$, the associated scalar score is
\[
\rho(\delta)=\inf\{r:\delta\in\D_r\}.
\]
According to Assumption~\ref{ass:classical_methods}, we can compute the multivariate confidence region for the difference between the target functional evaluated under synthetic distribution and the true distribution, 
\[
K_j :=\Delta^{\theta_j,\alpha_2}(S_j,\tilde S_j),
\]
then define
\[
R_j
=
R(K_j)
=
\inf\{r:K_j\subseteq\D_r\}.
\]
Whenever $\Delta_j\in K_j$, we have $\rho(\Delta_j)\leq R_j$. Thus, calibration based on exchangeability of the scalar scores $R_j$ gives a valid convex prediction region for the unobserved current gap $\Delta_{T+1}$.

For sets $A,B\subseteq\R^d$, write
\[
A\oplus B=\{a+b:a\in A,\; b\in B\}
\]
for their Minkowski sum. The multidimensional procedure is given in Algorithm~\ref{alg:multi_dimensional_targets}. As before, $R_{(k)}$ denotes the $k$-th order statistic of $R_1,\dots,R_T$, with the convention $R_{(T+1)}=\infty$.

\begin{algorithm}[t]
\caption{Inference via task exchangeability for multi-dimensional targets}
\label{alg:multi_dimensional_targets}
\begin{algorithmic}[1]
\Require current task $\cT^*$, historical task--dataset pairs $(\cT_1,S_1),\dots,(\cT_T,S_T)$, synthetic data generator $\cG$, synthetic sample size $N$, error levels $\alpha_1,\alpha_2,\alpha_3\in(0,1)$, nested convex family $\{\D_r:r\in\R\}$
\State Draw $N$ synthetic samples for the current task: $\tilde S\stackrel{\mathrm{i.i.d.}}{\sim}\tilde\cP$, where $\tilde\cP=\cG_{\cT^*}$
\State Draw $N$ synthetic samples for each historical task: $\tilde S_j\stackrel{\mathrm{i.i.d.}}{\sim}\tilde\cP_j$, where $\tilde\cP_j=\cG_{\cT_j}$, for all $j\in[T]$
\State Compute a confidence region for the current synthetic-data target $\tilde\theta$:
$\widetilde {\mathrm{CI}}=\mathrm{CI}^{\theta^*,\alpha_1}(\tilde S)$
\State For each historical task $j\in[T]$, compute a confidence region for the gap $\Delta_j$:
$K_j=\Delta^{\theta_j,\alpha_2}(S_j,\tilde S_j)$
\State Compute the scalar score $R_j=\inf\{r:K_j\subseteq\D_r\}$ for each $j\in[T]$
\State Let $k=\lceil (T+1)(1-\alpha_3)\rceil$ and let $\hat R=R_{(k)}$
\Ensure Confidence region for $\theta^*$:
$\ci=\widetilde{\mathrm{CI}} \oplus\D_{\hat R}$
\end{algorithmic}
\end{algorithm}

\begin{theorem}
\label{thm:multi_dimensional_targets}
Suppose Assumptions~\ref{ass:task_exchangeability} and~\ref{ass:classical_methods} hold. Then, the confidence region $\ci$ output by Algorithm~\ref{alg:multi_dimensional_targets} satisfies
\[
P(\theta^*\in\ci)\geq 1-(\alpha_1+\alpha_2+\alpha_3).
\]
\end{theorem}

The choice of nested family determines the shape of the resulting confidence region. A useful default is a translated and scaled convex family. Let $B\subseteq\R^d$ be a convex compact set containing the origin in its interior, and define
\[
\D_r=c+rB,\qquad r\geq 0,
\]
with $\D_r=\emptyset$ for $r<0$. If $\gamma_B(x)=\inf\{r:x\in rB\}$, then
\[
R_j
=
\sup_{\delta\in K_j}\gamma_B(\delta-c).
\]
The set $B$ can be the Euclidean unit ball, or a $d$-dimensional rectangle, for example.


\subsection{Inference on the finite-sample target \texorpdfstring{$\theta^*(S)$}{target}}
\label{sec:finite_sample_target}
 
Algorithm~\ref{alg:main_algo} targets the population estimand $\theta^*=\theta^*(\cP^*)$. Sometimes, however, one might be satisfied with the finite-sample target $\theta^*(S)$, the functional evaluated on the (counterfactual) real dataset $S=S_{T+1}$ for the current task. This is the only quantity against which coverage can actually be checked; indeed, when we estimate coverage on real datasets in our experiments, we treat the held-out task's empirical value $\theta^*(S)$ as the ground truth, precisely because the population value $\theta^*(\cP^*)$ is never observed. We now show that inference on $\theta^*(S)$ admits a simpler and tighter procedure than Algorithm~\ref{alg:main_algo}.
 
The key observation is that, for the finite-sample target, the relevant discrepancy is between two \emph{observable} sample statistics rather than between two population functionals. For each task $j$, define the finite-sample gap
\[
\hat \Delta_j = \theta_j(S_j)-\theta_j(\tilde S_j),\qquad j\in[T],
\]
together with the unobserved current-task gap $\hat \Delta_{T+1}=\theta^*(S)-\theta^*(\tilde S)$. Unlike the population gap $\Delta_j=\theta_j(\cP_j)-\theta_j(\tilde\cP_j)$ calibrated in Algorithm~\ref{alg:main_algo}, each $\hat \Delta_j$ is computed directly from the data, with no confidence interval required. Two simplifications follow. First, the synthetic-data estimate $\theta^*(\tilde S)$ is observed exactly and enters as a point rather than through a confidence interval, so the $\alpha_1$ step of Algorithm~\ref{alg:main_algo} is no longer needed. Second, because the historical gaps $\hat \Delta_j$ are observed, no per-task gap interval is required, and the $\alpha_2$ (per-task gap interval) and $\alpha_3$ (application of exchangeability) steps of Algorithm~\ref{alg:main_algo} collapse into a single exchangeability-based calibration using the $\hat \Delta_j$. We therefore use the entire budget $\alpha$ on one application of exchangeability. The procedure is stated in Algorithm~\ref{alg:finite_sample_target}.
 
As before, we write $\hat \Delta_{(k)}$ for the $k$-th order statistic of $\hat \Delta_1,\dots,\hat \Delta_T$, with the conventions $\hat \Delta_{(0)}=-\infty$ and $\hat \Delta_{(T+1)}=+\infty$.
 
\begin{algorithm}[t]
\caption{Inference on the finite-sample target via task exchangeability}
\label{alg:finite_sample_target}
\begin{algorithmic}[1]
\Require current task $\cT^*$, historical task--dataset pairs $(\cT_1,S_1),\dots,(\cT_T,S_T)$, synthetic data generator $\cG$, synthetic sample size $N$, error level $\alpha\in(0,1)$
\State Draw $N$ synthetic samples for the current task: $\tilde S\stackrel{\mathrm{i.i.d.}}{\sim}\tilde\cP$, where $\tilde\cP=\cG_{\cT^*}$
\State Draw $N$ synthetic samples for each historical task: $\tilde S_j\stackrel{\mathrm{i.i.d.}}{\sim}\tilde\cP_j$, where $\tilde\cP_j=\cG_{\cT_j}$, for all $j\in[T]$
\State Compute the synthetic-data estimate for the current task: $\theta^*(\tilde S)$
\State For each historical task $j\in[T]$, compute the finite-sample gap:
$\hat \Delta_j=\theta_j(S_j)-\theta_j(\tilde S_j)$
\State Let $\hat \Delta^L=\hat \Delta_{(k_L)}$ and $\hat \Delta^U=\hat \Delta_{(k_U)}$, where $k_L=\lfloor (T+1)\frac{\alpha}{2}\rfloor$ and $k_U=\lceil (T+1)(1-\frac{\alpha}{2})\rceil$
\Ensure Confidence interval for $\theta^*(S)$:
$\ci=[\theta^*(\tilde S)+\hat \Delta^L,\ \theta^*(\tilde S)+\hat \Delta^U]$
\end{algorithmic}
\end{algorithm}
 
\begin{theorem}
\label{thm:finite_sample_target}
Suppose Assumption~\ref{ass:task_exchangeability} holds. Then, the confidence interval $\ci$ output by Algorithm~\ref{alg:finite_sample_target} satisfies
\[
P\left(\theta^*(S)\in\ci\right)\geq 1-\alpha.
\]
Moreover, if the finite-sample gaps $\hat\Delta_1,\dots,\hat\Delta_{T+1}$ are almost surely distinct, then
\[
P\left(\theta^*(S)\in\ci\right)
<
1-\alpha+\frac{2}{T+1}.
\]
\end{theorem}

Therefore, under a mild condition, coverage of the finite-sample target is tight around $1-\alpha$, and the slack decreases with the number of historical tasks.

\subsection{What if we have some real data from the target task?}
\label{sec:some_real_data}

So far, we have focused on the setting where no real data is available for the current task. In some cases, however, one may have a small real sample from the target distribution $\cP^*$. This creates a natural tradeoff. If the synthetic data is poorly specified, then inference should rely primarily on the real sample. Conversely, if the synthetic data is accurate for the current task, then it may be beneficial to rely more heavily on the synthetic sample, which can be generated at much larger scale.

A simple way to obtain this adaptivity is to construct two confidence intervals (at a discounted error level) and intersect them. First, using the available real data from the current task, construct a confidence interval $\CI^{\alpha_1}_{\mathrm{real}}$ satisfying
\[
P(\theta^*\in \CI^{\alpha_1}_{\mathrm{real}})\geq 1-\alpha_1.
\]
Then, using the synthetic-data procedure based on task exchangeability, construct a confidence interval $\CI^{\alpha_2}_{\mathrm{syn}}$ satisfying
\[
P(\theta^*\in \CI^{\alpha_2}_{\mathrm{syn}})\geq 1-\alpha_2.
\]
Finally, we take
\[
\overline{\CI}
=
\CI^{\alpha_1}_{\mathrm{real}}
\cap
\CI^{\alpha_2}_{\mathrm{syn}}.
\]
By a union bound,
\[
P(\theta^*\notin \overline{\CI})
\leq
P(\theta^*\notin \CI^{\alpha_1}_{\mathrm{real}})
+
P(\theta^*\notin \CI^{\alpha_2}_{\mathrm{syn}})
\leq
\alpha_1+\alpha_2.
\]
Thus, as long as $\alpha_1+\alpha_2=\alpha$, the intersected interval $\overline{\CI}$ has coverage at least $1-\alpha$.

This construction gives a simple way to interpolate between the real-data and synthetic-data intervals. One can set $\alpha_1=\lambda\alpha$ and $\alpha_2=(1-\lambda)\alpha$ for some $\lambda\in(0,1)$. Larger values of $\lambda$ allocate more error budget to the real-data interval, making that interval less conservative, while smaller values of $\lambda$ allocate more error budget to the synthetic-data interval. 

In practice, $\lambda$ can be fixed in advance or tuned using a held-out sample that is independent of the data used to construct the final intervals. A natural choice is to select the value of $\lambda$ that minimizes the length of the resulting intersection interval on the held-out sample. Conditional on the chosen $\lambda$, the final interval remains valid by the same union-bound argument, provided the final confidence intervals are constructed on data independent of the tuning step. More sophisticated approaches such as cross-fitting would lead to shorter intervals, but their validity would require additional stability or asymptotic arguments. We defer a thorough investigation of how to combine real and synthetic data for the target task to future work.

We note that the setup considered in this section is analogous to the setup studied by \citet{bashari2025statistical}. Their method, GESPI, includes a guardrail based on performing inference with real data only at a slightly more liberal error level $\alpha + \epsilon$, where $\alpha$ is the target level and $\epsilon$ is a small additional factor. As a result, GESPI is in the best case as powerful as standard inference on real data at level $\alpha + \epsilon$. Our method, on the other hand, can be far more powerful if the synthetic data generator is good, however it crucially relies on the assumption of task exchangeability (or, more generally, its coverage degrades with distance to exchangeability, as in Theorem \ref{thm:beyond_exchangeability}).

\section{Experiments}

We evaluate our approach to valid inference with synthetic data on several real-world datasets from public opinion research and AI evaluation with a reward-model-based autorater. For each experiment, we justify the existence of historical tasks that can be used to calibrate the errors of the synthetic data generator. We use standard confidence intervals based on the central limit theorem to satisfy the requirement of having access to classical inference methods, as stated in Assumption \ref{ass:classical_methods}.

We visualize the computed confidence intervals and report their mean width and estimated coverage. To estimate coverage on real datasets, we merely average over the tasks and treat the finite-sample target computed on the held-out task, i.e. $\theta^*(S)$, as $\theta^*(\cP^*)$. Note that this leads to a somewhat conservative coverage estimate, seeing that Section \ref{sec:finite_sample_target} gives a more powerful procedure for estimating $\theta^*(S)$. We complement our real-data findings with simulated-data experiments in order to be able to know $\theta^*(\cP^*)$ and evaluate coverage exactly.

All code and data needed to reproduce the results in this paper are
available at
\url{https://github.com/CarrieTan13/synthetic-data-inference}.

\subsection{Silicon sampling for American National Election Studies (ANES)}

We consider the use of synthetically generated respondents in public opinion surveys. Surveys are the standard instrument for measuring social and political attitudes, but fielding nationally representative surveys is slow and expensive. A growing line of work therefore asks whether large language models (LLMs) can stand in for human respondents, simulating the answers a person with a given demographic and political profile would give~\cite{argyle2023out}. This practice, known as ``silicon sampling,'' promises faster and cheaper measurement, but the simulated responses can, of course, be biased and need not track real public opinion.

Our target estimand is the population mean of an American National Election Studies (ANES) feeling
thermometer, in which respondents rate a social or political group on a
$0$--$100$ scale ($0$ maximally unfavorable, $100$ maximally favorable). For a
target group and a partisan subgroup of the electorate, the estimand is the
mean rating that the target group receives within that subgroup.

We treat each target group--respondent partisanship pair as one task. Historical tasks correspond to feeling thermometer scores measured in an earlier survey wave, for which both real survey responses and LLM-simulated responses are available; the
target task asks for inference on a thermometer score in a later wave, for which we only have simulated responses. For each historical task $j$, the
real-data estimand $\theta_j$ is the ANES mean and the synthetic-data estimand
$\tilde\theta_j$ is the corresponding LLM-simulated mean, so the gap
$\Delta_j = \theta_j - \tilde\theta_j$ measures the simulation bias of the LLM
for task $j$.

We use the silicon samples collected by \citet{bisbee2024synthetic}, generated by GPT-3.5 based on the $2016$ and $2020$ waves of the American National Election Studies (ANES). Synthetic responses are generated by prompting the model to adopt each respondent's demographic and political profile and return a $0$--$100$ rating. Combining the eleven thermometer targets with respondent partisanship (Democrat, Republican, Independent) yields $T = 33$ tasks per wave; we calibrate on the $2016$ tasks and predict the estimands for the corresponding $2020$ tasks, applying Algorithm~\ref{alg:main_algo} to each target and comparing the resulting intervals to those obtained by treating the simulated ratings as real survey responses. Because each respondent contributes a single LLM rating, the synthetic sample size matches the real sample size in every task, $N_j = n_j$; the $33$ historical tasks from 2016 contain between $435$ and $780$ respondents (median $762$), and the $33$ target tasks from 2020 between $657$ and $1,795$ respondents (median $1,596$). We set $\alpha=0.15$.

Figure~\ref{fig:anes_main} reports the resulting intervals across all
$33$ tasks from 2020 calibrated on $T = 33$ tasks from 2016. The task-exchangeability intervals from Algorithm~\ref{alg:main_algo} cover the
held-out $2020$ thermometer mean in $97\%$ of tasks, with a median interval width
of $29.8$ points on the $0$--$100$ scale. The naive synthetic-only intervals---which treat the
GPT-3.5--simulated ratings as if they were real survey responses---are far too
narrow and systematically biased, covering the truth in only $3\%$ of tasks. The
simulated ratings vary far less than real responses, so the naive intervals
concentrate tightly around a biased center and almost never contain the true means from 2020. This is consistent with findings by \citet{bisbee2024synthetic}.

\begin{figure}[t!]
    \centering
    \includegraphics[width=0.9\linewidth]{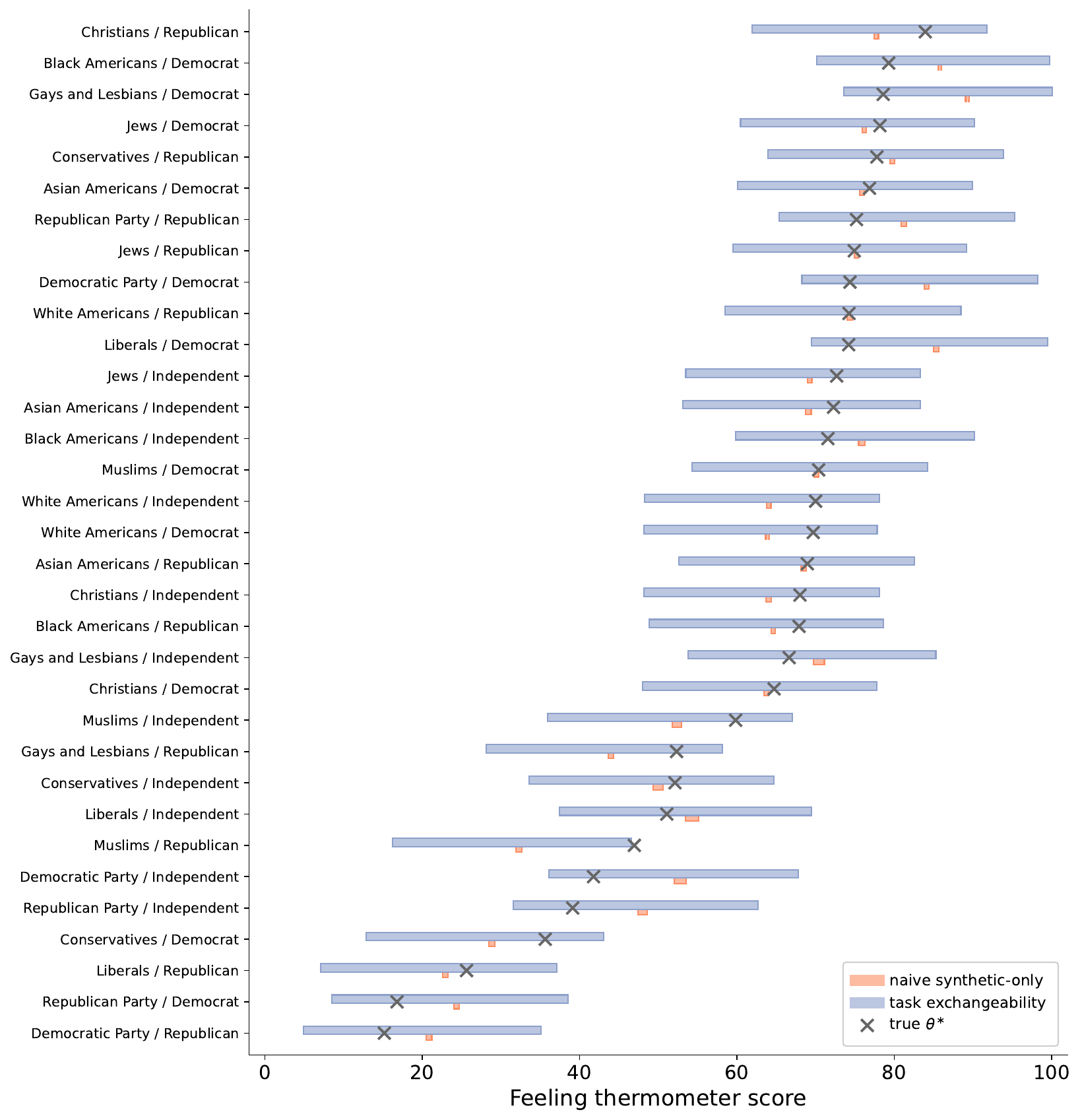}
    \caption{\textbf{Inference on ANES feeling-thermometer scores via task exchangeability.}
Each row corresponds to a task defined by a target group and respondent subgroup, with the estimand equal to the average ANES feeling-thermometer score on a 0--100 scale.}
    \label{fig:anes_main}
\end{figure}

\begin{figure}[htp!]
    \centering    \includegraphics[width=1\linewidth]{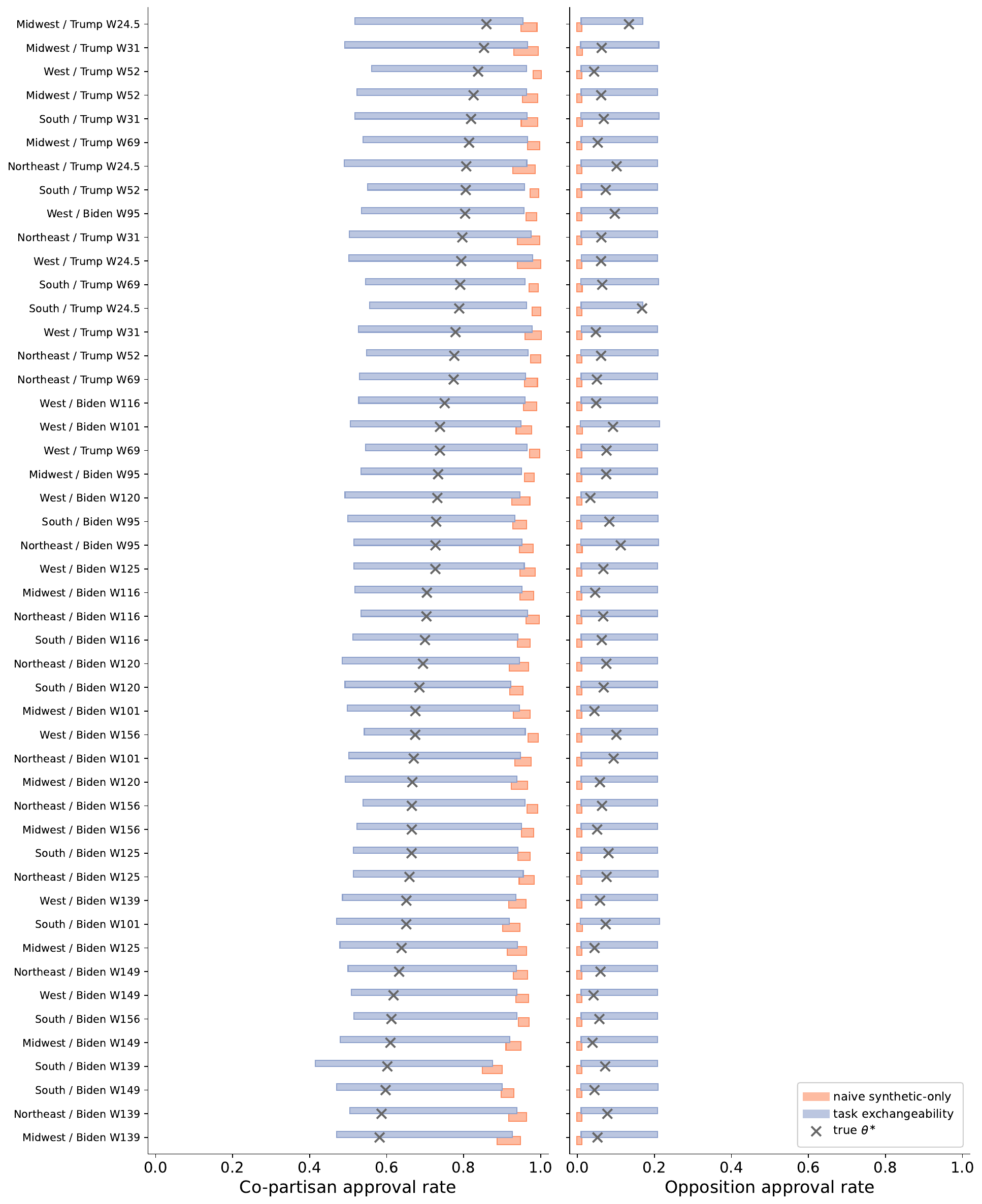}
    \caption{\textbf{Inference on presidential approval based on Pew ATP surveys via task exchangeability.}
Each row corresponds to a field date and census region, with the two-dimensional estimand equal to the average approval among respondents who share the president's party (left column) and the average approval among respondents from the opposing party (right column).}
    \label{fig:forest_pew_gpt_4o_multidim}
\end{figure}

\subsection{Silicon sampling for Pew American Trends Panel (ATP)}

We turn to a second public-opinion task: forecasting
presidential approval from the Pew Research Center's American Trends Panel
(ATP). Across many survey waves Pew asks the same question---\emph{``Do you approve
or disapprove of the way [the sitting president] is handling his job as
president?''}---coded as \texttt{POL1JB} for Joe Biden and
\texttt{POL1DT} for Donald Trump, with answers chosen from \emph{approve} and 
\emph{disapprove}. We use $12$ field dates spanning February~$2017$ to
October~$2024$ ($4$ Trump-era waves and $8$ Biden-era waves)\footnote{Datasets are publicly available from the Pew Research Center
American Trends Panel, \url{https://www.pewresearch.org/american-trends-panel-datasets/}. The
$12$ waves we use and their field dates are W24.5 (Feb 28--Mar 12, 2017),
W31 (Jan 29--Feb 13, 2018), W52 (Jul 22--Aug 4, 2019), W69 (Jun 16--22, 2020),
W95 (Sep 13--19, 2021), W101 (Jan 10--17, 2022), W116 (Oct 10--16, 2022),
W120 (Jan 18--24, 2023), W125 (Mar 27--Apr 2, 2023), W139 (Nov 27--Dec 3, 2023),
W149 (Jul 1--7, 2024), and W156 (Sep 30--Oct 6, 2024).} and take the
estimand to be the survey-weighted proportion of respondents who approve. Because
Pew releases survey weights depending on respondent demographics, each estimand is a weighted
approval proportion.
 
Presidential approval is sharply polarized: co-partisans of the sitting president
approve at roughly $70$--$85\%$, while the opposing party approves at only
$5$--$15\%$, with almost no mass in between. A single approval rate per region
therefore blends two nearly disjoint regimes. We instead make the target
two-dimensional: for each (field date, census region) cell we infer the pair
\[
  \theta_j=\bigl(\theta_j^{\mathrm{co}},\;\theta_j^{\mathrm{opp}}\bigr),
\]
where $\theta_j^{\mathrm{co}}$ is the approval rate among respondents who share
the president's party (Democrats under Biden, Republicans under Trump) and
$\theta_j^{\mathrm{opp}}$ is the approval rate among the opposing party. A task is
thus one (field date, census region) cell, and we apply the multidimensional
procedure of Algorithm~\ref{alg:multi_dimensional_targets}, forming a rectangular
joint confidence region whose two coordinates are calibrated separately and
combined through a Bonferroni correction.
 
Unlike the ANES experiment, where we use the silicon samples released by
\citet{bisbee2024synthetic}, we generate the synthetic Pew responses ourselves.
For every panelist we query GPT-4o through the API, supplying the
respondent's demographic and political profile (census region, age, gender,
education, race and ethnicity, religion, income, and party identification) and
recording the model's approve/disapprove answer. We take four independent
generations per respondent and average them to form a synthetic approval
indicator; the synthetic estimand $\tilde\theta_j$ is the weighted mean of these
indicators, and the gap $\Delta_j=\theta_j-\tilde\theta_j$ measures GPT-4o's
simulation bias for cell $j$.
 
We first evaluate by leaving out one task at a time: each of the $48$ tasks ($12$ field dates $\times\
4$ census regions) is held out in turn and its two-dimensional target is calibrated
via the remaining $T = 47$. This differs from the ANES experiment, where calibration
and target tasks come from two distinct survey years.
Per coordinate, the held-out cell contains $n_j$ real
respondents---co-partisan cells have a median of $568$ respondents (range
$255$--$1,898$) and opposition cells a median of $515$ (range $243$--$2,058$); the synthetic data contains
$N_j=n_j$ samples averaged on $4\,n_j$ GPT-4o generations. In total the
design uses roughly $66{,}000$ respondent-level real--synthetic pairs. Coverage of
a held-out cell is assessed against its real-correspondent weighted approval rate, as
described above.
 
Figure~\ref{fig:forest_pew_gpt_4o_multidim} reports the result at level
$\alpha=0.2$. The two panels show the intervals of the rectangular
region for the co-partisan and opposition coordinates, each sorted by its true
rate; the polarization is visible as the separation between the high
($\theta^{\mathrm{co}}\!\approx\!0.6$--$0.85$) and low
($\theta^{\mathrm{opp}}\!\approx\!0.03$--$0.17$) clusters. The
task-exchangeability region (Algorithm~\ref{alg:multi_dimensional_targets})
attains joint coverage for all $48$ tasks, with mean marginal
widths of about $0.44$ (co-partisan) and $0.20$ (opposition). In contrast, the
naive synthetic-only rectangle---tight intervals centered on the GPT-4o means---covers zero tasks: the simulated approval is systematically biased, with GPT-4o
overstating co-partisan approval and understating opposition approval, so its
narrow intervals never contain the truth.

To mirror the potential application more closely---calibrating on past waves and
forecasting future ones---we also split the tasks chronologically: we calibrate
on the $T = 40$ tasks from the ten earliest field dates and predict the eight tasks
from the two most recent waves (W149 and W156), each crossed with the four regions. 
Figure~\ref{fig:pew_temporal} shows the result at $\alpha=0.2$: the rectangular
task-exchangeability intervals from Algorithm~\ref{alg:multi_dimensional_targets} contain the realized approval rate in every
one of the eight forecast cells on both coordinates, whereas the naive
synthetic-only intervals---centered near $0.95$ for co-partisan and near $0$ for
opposition approval---miss the truth in all of them.

\begin{figure}[t]
    \centering
    \includegraphics[width=1\linewidth]{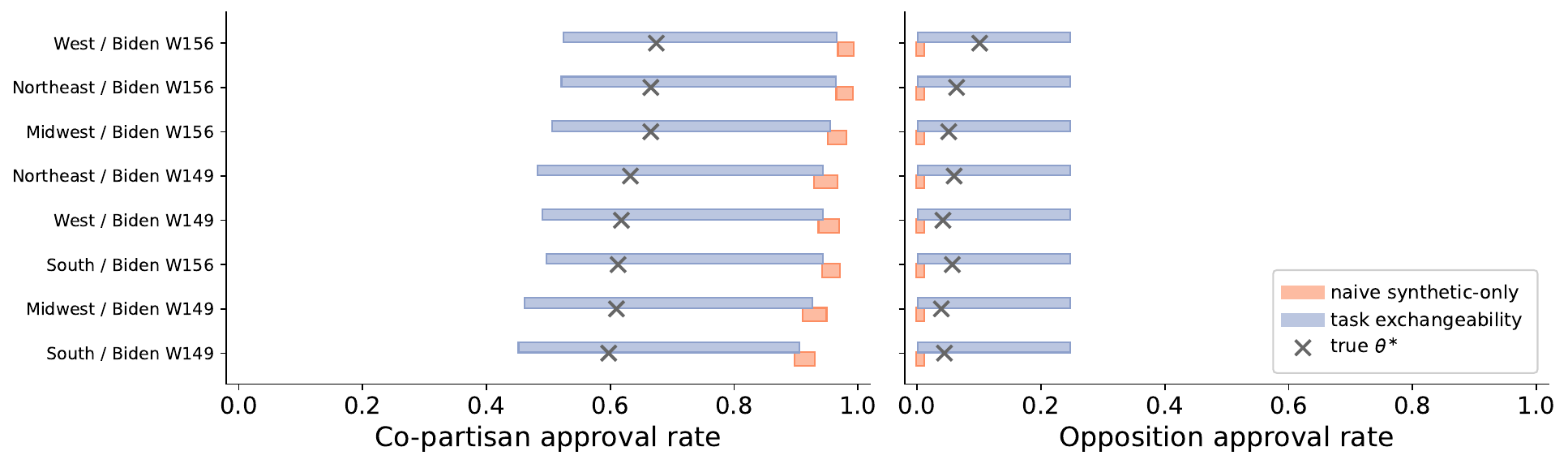}
    \caption{\textbf{Inference on presidential approval based on Pew ATP surveys via task exchangeability.}
Each row corresponds to a field date and census region, with the two-dimensional estimand equal to the average approval among respondents who share the president's party (left column) and the average approval among respondents from the opposing party (right column). Only tasks from earlier years are used for calibration.}
    \label{fig:pew_temporal}
\end{figure}

\begin{figure}[t!]
    \centering
    \includegraphics[width=0.9\linewidth]{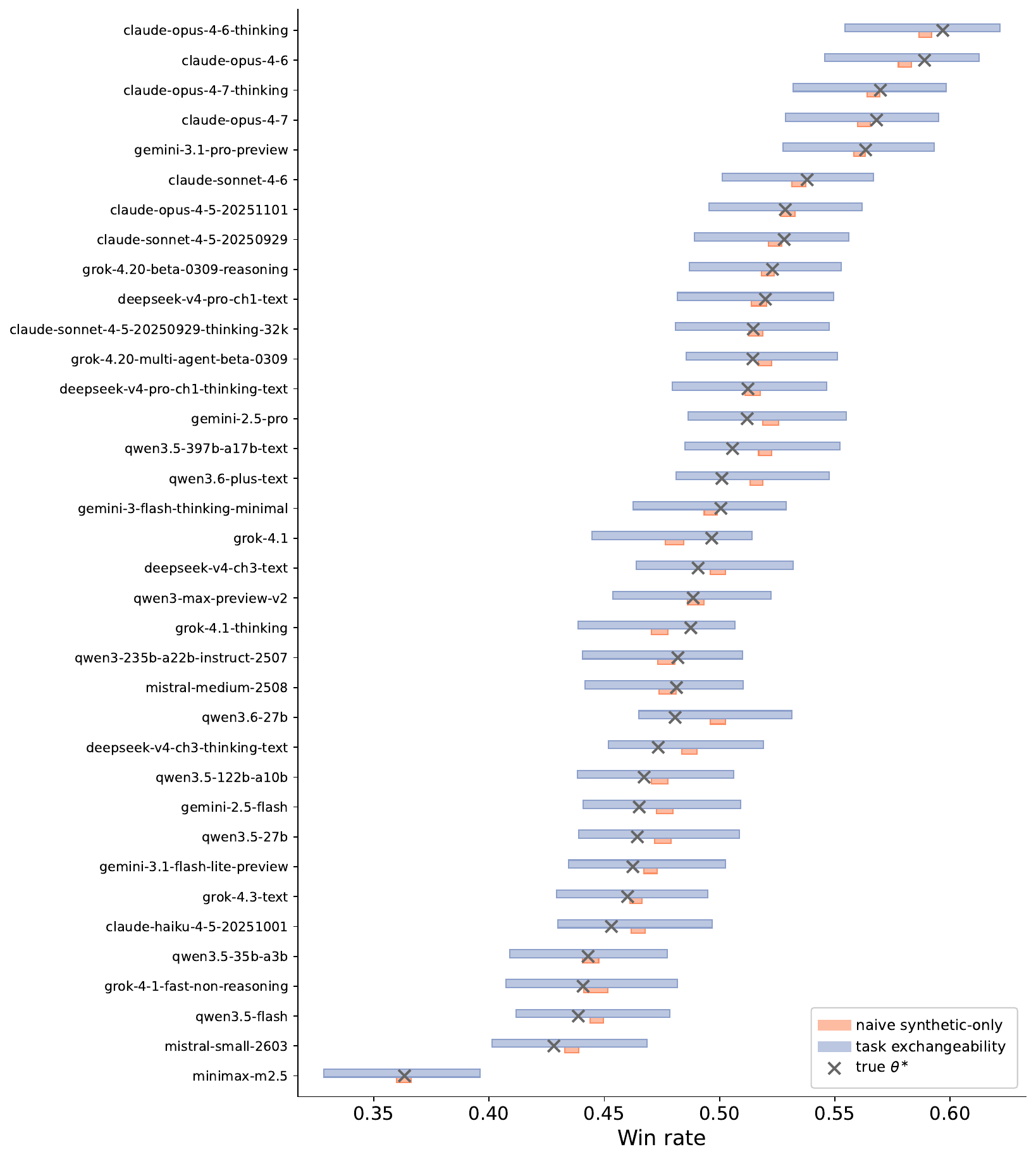}
    \caption{\textbf{Inference on Arena model win rates via task exchangeability.}
Each row corresponds to an AI model, with the estimand equal to the win rate of that model against the competing pool of models.}
\label{fig:forest_autorater_naive_vs_alg1}
\end{figure}

\subsection{AI model evaluation with autoraters}

We next consider AI model evaluation. Human preference evaluation is a central way to compare AI models, but collecting fresh pairwise preference data for every new model is costly and slow. This creates a practical tension in settings such as Arena \cite{chiang2024chatbot}: when a new model is released, one would like to obtain an initial evaluation quickly, while human preference data may only arrive gradually over time. As a result, ``autoevaluation'' has become increasingly common: rather than asking human annotators to compare model outputs, one uses an autorater, such as a reward model or LLM judge \cite{zheng2023judging}, to predict which response a human would prefer. Autoevaluation can therefore provide a quick preliminary evaluation of a new model, which can later be refined as human comparisons are collected. This makes evaluation substantially faster and cheaper, but introduces potential biases and misalignment with human judgment~\cite{ye2025justice}.


We consider two target estimands. The first is the win rate: for a model $m$,
$\theta_m$ is the probability that a human evaluator prefers the response from $m$ to
that of the opponent, when the prompt and opponent are drawn a particular reference distribution. Ties are counted as $\tfrac12$. The second estimand is the model's Bradley-Terry (BT)
score~\cite{bradley1952rank}, in which case $\theta_m$ is the ``strength'' of model
$m$ in a BT regression fit to the same comparisons. The BT solution is made identifiable by anchoring the score of a reference model to a fixed value; we anchor the score of one model to zero.

We treat each model evaluation as one task. In a real deployment, historical tasks correspond to previously evaluated models for which both human preference data and autorater-based evaluations are available. The target task corresponds to a newly released model for which we wish to infer the human-preference win rate or BT coefficient before enough human comparisons have been collected. For each historical model $j$, the real-data estimand $\theta_j$ is its human-preference win rate or BT coefficient, while the synthetic-data estimand $\tilde \theta_j$ is the corresponding estimand induced by the autorater. The gap $\Delta_j = \theta_j - \tilde \theta_j$ measures the evaluation bias of the autorater for model $j$. Thus, historical human evaluations are used to calibrate the autorater's bias, enabling quick inference for a new model while human data is accumulated more slowly over time. We evaluate the method by leave-one-out: each model is treated as the target task in turn, while all remaining models serve as the historical exchangeable tasks. 

For the experiment, we use prompts, human votes, and autorater votes from Arena \cite{chiang2024chatbot}. The autorater votes were generated using a reward model trained by Arena. The autorater takes a prompt and two model responses, and returns a predicted probability that a human would prefer the first response. The design is paired: every prompt yields both a human vote and an autorater vote
on the same pair of responses, so
$N_j=n_j$. The data comprises $280{,}737$ human--autorater comparisons among $74$ models, of which $73$ are evaluated as targets; the per-model sample size $n_j$ has a median of $7{,}918$ (range $2{,}865$--$20{,}324$). We apply Algorithm \ref{alg:main_algo} and Algorithm \ref{alg:finite_sample_target} across models, comparing the resulting intervals to those computed by treating autorater votes as real human-preference votes. 

Figure~\ref{fig:forest_autorater_naive_vs_alg1} reports inference on the win rate across all
$73$ evaluated models at level $\alpha = 0.1$. The task-exchangeability intervals from Algorithm~\ref{alg:main_algo}
cover the held-out human win rate for all models, with a median interval
width of $0.067$ on the win-rate scale. By contrast, the naive intervals---which treat the autorater votes as if they were real human preferences---are sharply
overconfident, covering the true win rate for only $19\%$ of models.

Figure~\ref{fig:forest_autorater_alg1_vs_onepiece} compares inference on $\theta^*(\cP^*)$ and inference on $\theta^*(S)$ in this problem. We apply Algorithm~\ref{alg:main_algo} and
Algorithm~\ref{alg:finite_sample_target} for the two targets, respectively. Both achieve at least the nominal
$90\%$ coverage: Algorithm~\ref{alg:main_algo} covers all of the held-out
targets, while Algorithm~\ref{alg:finite_sample_target} covers $91.8\%$. The latter value is
notable because it lands inside the narrow band
$\bigl[\,1-\alpha,\ 1-\alpha+\tfrac{2}{T+1}\,\bigr)=[0.9,\,0.927)$ guaranteed
by Theorem~\ref{thm:finite_sample_target}. This confirms the two-sided nature of
the guarantee: with $T=73$ tasks the finite-sample procedure is not only valid
but tightly calibrated, neither undercovering below $1-\alpha$ nor overcovering
by more than $\tfrac{2}{T+1}\approx 2.7$ percentage points. As expected, the intervals for the finite-sample target are also less than half as wide, with a median width of
$0.032$ against $0.067$ on the win-rate scale. The gain comes from the structure
of the finite-sample target: because $\theta^*(\tilde S)$ is observed exactly and
the historical gaps $\hat\Delta_j$ require no per-task confidence interval, the
entire error budget is spent on a single exchangeability step rather than split
across the three budgets $\alpha_1,\alpha_2,\alpha_3$ of
Algorithm~\ref{alg:main_algo}.

\begin{figure}[htp!]
\centering
\includegraphics[width=1\linewidth]{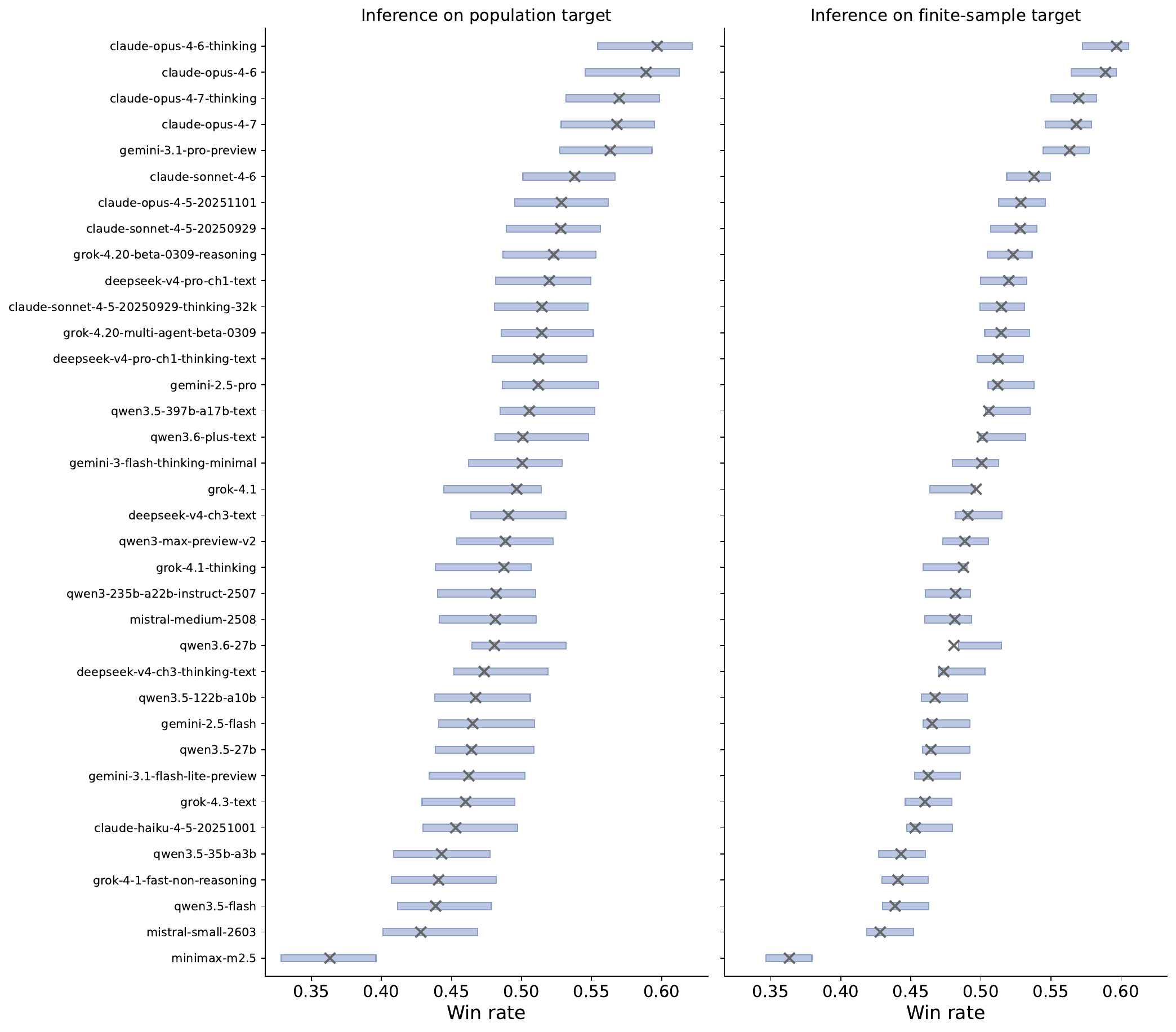}
    \caption{\textbf{Inference on population win rate (left) vs finite-sample win rate (right) via task exchangeability.}
Each row corresponds to an AI model, with the estimand equal to the win rate of that model against the competing pool of models. Note that the cross ($\times$) marks the finite-sample target in both cases, since we cannot compute the population target from real data.}
\label{fig:forest_autorater_alg1_vs_onepiece}
\end{figure}

Finally, Figure~\ref{fig:forest_autorater_bt} reports inference on the BT score.
Here the estimand for each model is its Bradley-Terry strength, obtained by inserting
the model into a regression fit on all existing comparisons of remaining models, with a
fixed reference model anchored at zero. The $44$ earliest-released models remain in the
regression at all times but are never treated as tasks (in practice these correspond
to existing public models); we run leave-one-out over the $30$ newly published
models. We target the finite-sample BT score via
Algorithm~\ref{alg:finite_sample_target}, comparing against the naive intervals
that treat the autorater votes as if they were real human preferences. For
readability the figure displays the $17$ held-out models belonging to frontier
model families, but the reported coverage and interval widths are computed over all
$30$ held-out models. At level $\alpha = 0.1$, Algorithm~\ref{alg:finite_sample_target}
covers the held-out finite-sample target for $93.3\%$ of models, with a median
interval width of $0.15$ on the log-strength scale. As in the win-rate
experiment, this coverage lands inside the band
$[1-\alpha, 1-\alpha+\frac{2}{T+1})=[0.9,0.967)$ guaranteed by
Theorem~\ref{thm:finite_sample_target}, here with $T=29$ historical tasks: the
procedure is not only valid but tightly calibrated, neither undercovering below
$1-\alpha$ nor overcovering by more than $\tfrac{2}{T+1}\approx 6.7$ percentage
points. By contrast, the naive intervals undercover, capturing the true BT score for
only about $73\%$ of models, well short of the nominal $90\%$.

\begin{figure}[t!]
\centering
\includegraphics[width=0.6\linewidth]{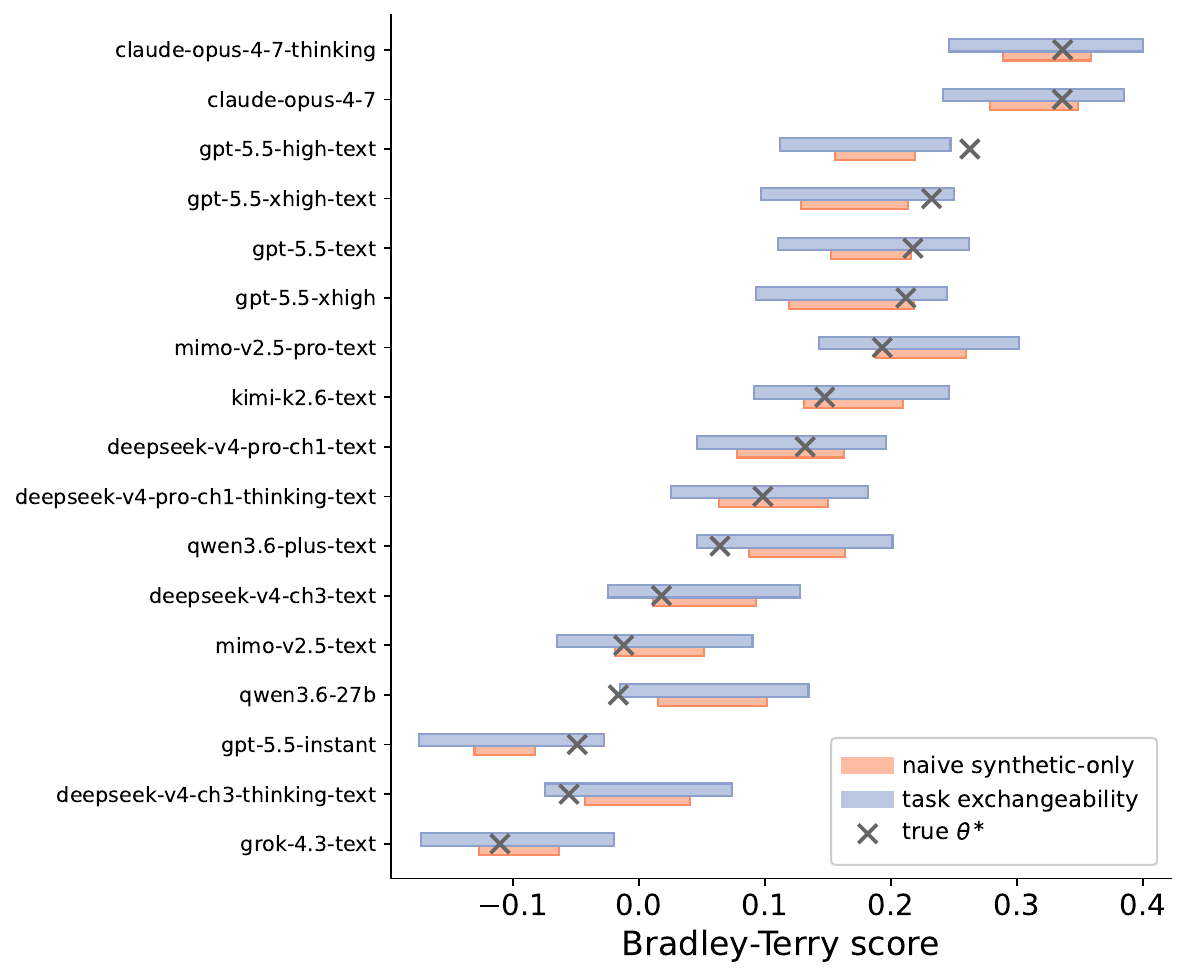}
\caption{\textbf{Inference on finite-sample Bradley-Terry score via task exchangeability.}
Each row corresponds to an AI model, with the estimand equal to its Bradley-Terry
score.}
\label{fig:forest_autorater_bt}
\end{figure}

\subsection{Simulated data}
\label{sec:synthetic_data}

In the applications above we do not have access to the true target
$\theta^*(\cP^*)$; we instead use the finite-sample quantity $\theta^*(S)$
computed on the held-out task as a proxy, which yields a somewhat conservative
coverage estimate. To evaluate coverage exactly, we complement those
experiments with simulated data, where the data-generating process, and hence
$\theta^*(\cP^*)$, is known. 

We mimic the Pew public opinion setting in which each task is a binary survey item and
$\theta_j$ is a population proportion of positive outcomes. For a pool of $T+1$ tasks, we draw, independently across $j \in [T+1]$,
\[
  p_j \sim \mathrm{Beta}(2,2), \qquad
  \varepsilon_j \sim \cN(\mathrm{bias},\,\tau^2), \qquad
  \tilde p_j = \mathrm{clip}(p_j+\varepsilon_j,\,0,\,1),
\]
and then sample a real dataset and a synthetic dataset
\[
  S_j \sim \mathrm{Bernoulli}(p_j)^{\otimes n_j}, \qquad
  \tilde S_j \sim \mathrm{Bernoulli}(\tilde p_j)^{\otimes N_j}.
\]
The target is the mean, $\theta_j=p_j$, the synthetic estimand is
$\tilde\theta_j=\tilde p_j$, and the simulation bias is
$\Delta_j=\theta_j-\tilde\theta_j$ with variance $\tau^2$. Because the tuples
$(p_j,\varepsilon_j,S_j,\tilde S_j)$ are i.i.d.\ across $j$, the sequence is
exchangeable, satisfying Assumption~\ref{ass:task_exchangeability}.
We use $\mathrm{bias}=0.05$, $\tau=0.10$, and dataset sizes $n_j=1000$ and $N_j=2000$. We report results for $T\in\{40,100\}$ historical tasks. 
We evaluate the coverage and width of our method by resampling all data and tasks $1000$ times from scratch.

We use
$(\alpha_1,\alpha_2,\alpha_3)=(0.1,\,0.2,\,0.7)\,\alpha$, placing most of the
budget on the exchangeability step because the cross-task variance $\tau^2$ typically dominates
the width. Since this step needs $\alpha_3\ge 2/(T+1)$ to return finite
bounds ($0.049$ for $T=40$ and $0.020$ for $T=100$), for small $\alpha$ we clip $\alpha_3$ to equal $2/(T+1)$ and split the remaining budget $1:2$ between $\alpha_1$
and $\alpha_2$.
Since $\theta^*\in[0,1]$ is a proportion, we clip the confidence intervals to $[0,1]$.

We report the width and coverage for varying $\alpha$ in Table~\ref{tab:synthetic_alg1}. The method attains coverage above the nominal $1-\alpha$
in every cell, confirming validity; we observe that the coverage is conservative due to the error budget splitting. The
synthetic-only baseline, by contrast, covers the truth only $9\%$--$15\%$ of the
time.

Table~\ref{tab:synthetic_tau} varies $\tau$ at a
fixed level $\alpha=0.10$. Coverage stays at or above the nominal $0.90$
throughout (conservative, tightening as $\tau$ grows). The
width grows steadily with $\tau$---from $\approx0.26$ at $\tau=0.025$ to $\approx0.81$ at $\tau=0.30$. The synthetic-only
baseline does not benefit: its interval stays narrow ($\approx0.03$) and its
coverage, already poor, degrades further as the bias becomes more variable.

\begin{table}[t]
\centering
\begin{tabular}{c cc cc cc}
\toprule
 & \multicolumn{4}{c}{Inference via task exchangeability} & \multicolumn{2}{c}{Synthetic-only} \\
\cmidrule(lr){2-5}\cmidrule(lr){6-7}
 & \multicolumn{2}{c}{$T=40$} & \multicolumn{2}{c}{$T=100$} & & \\
\cmidrule(lr){2-3}\cmidrule(lr){4-5}
$\alpha$ & coverage & width & coverage & width & coverage & width \\
\midrule
0.05 & 0.997 & 0.550 & 0.996 & 0.575 & 0.146 & 0.037 \\
0.10 & 0.993 & 0.508 & 0.986 & 0.479 & 0.127 & 0.031 \\
0.15 & 0.977 & 0.437 & 0.977 & 0.429 & 0.105 & 0.027 \\
0.20 & 0.975 & 0.432 & 0.961 & 0.393 & 0.096 & 0.024 \\
\bottomrule
\end{tabular}
\caption{\textbf{Inference on mean with simulated data for varying $\alpha$ and $T$.} We run simulated-data experiments in order to get a reliable estimate of coverage. In all cases, we see that coverage exceeds the nominal value.}
\label{tab:synthetic_alg1}
\end{table}

\begin{table}[t]
\centering
\begin{tabular}{c cc cc cc}
\toprule
 & \multicolumn{4}{c}{Inference via task exchangeability} & \multicolumn{2}{c}{Synthetic-only} \\
\cmidrule(lr){2-5}\cmidrule(lr){6-7}
 & \multicolumn{2}{c}{$T=40$} & \multicolumn{2}{c}{$T=100$} & & \\
\cmidrule(lr){2-3}\cmidrule(lr){4-5}
$\tau$ & coverage & width & coverage & width & coverage & width \\
\midrule
0.025 & 1.000 & 0.261 & 1.000 & 0.249 & 0.086 & 0.032 \\
0.050 & 1.000 & 0.346 & 0.999 & 0.328 & 0.165 & 0.032 \\
0.100 & 0.993 & 0.508 & 0.986 & 0.479 & 0.127 & 0.031 \\
0.150 & 0.985 & 0.630 & 0.977 & 0.598 & 0.091 & 0.030 \\
0.200 & 0.978 & 0.714 & 0.969 & 0.681 & 0.054 & 0.029 \\
0.300 & 0.975 & 0.807 & 0.958 & 0.774 & 0.038 & 0.025 \\
\bottomrule
\end{tabular}
\caption{\textbf{Inference on mean with simulated data for varying $\tau$ and $T$.} We run simulated-data experiments in order to get a reliable estimate of coverage. The target coverage level is $0.9$. In all cases, we see that coverage exceeds the nominal value.}
\label{tab:synthetic_tau}
\end{table}

\newpage 
\section{Discussion}

The guarantees developed in this paper rely on the ability to identify historical tasks that are exchangeable, or approximately exchangeable, with the current task of interest. This is a substantive requirement, and it should not be viewed as automatic. In some applications, the assumption may be natural: for example, repeated survey waves measuring related quantities, or evaluations of AI models from the same or a similar family. In general, however, identifying an appropriate population of exchangeable tasks may be difficult. The framework is therefore most useful when the researcher can articulate, based on domain knowledge and study design, why the historical tasks provide relevant information about the current gap between the real-data and synthetic-data estimands.

It is also important to acknowledge that the coverage guarantees are \emph{marginal} over the task-generating process: they hold for a randomly drawn target task from a common population of tasks. They are not conditional on an arbitrary fixed task and all of its idiosyncratic features. This is analogous to conformal prediction, where finite-sample validity is marginal over the random test point rather than conditional on its realized covariates. Thus, task exchangeability provides a principled approach to statistical inference for a new task drawn from a relevant task population, but it is important to justify that population carefully.

\section*{Acknowledgements}

This work was supported by a Dieter Schwarz Early Career Catalyst Award.

\bibliographystyle{plainnat}
\bibliography{refs}

\appendix

\newpage

\section{Deferred proofs}

\subsection{Proof of Theorem \ref{thm:beyond_exchangeability}}

We first show that the calibrated upper endpoint satisfies
\[
P\left(\hat\Delta_{T+1}^U\leq \hat\Delta^U\right)
\geq
1-\frac{\alpha_3}{2}-\varepsilon_U.
\]
For any vector $v=(v_1,\dots,v_{T+1})\in\R^{T+1}$, define
\[
S_U(v)
=
\left\{
i\in[T+1]:
v_i>
Q_{1-\frac{\alpha_3}{2}}\left(\sum_{j=1}^{T+1}\bar w_j\delta_{v_j}\right)
\right\}.
\]
By the definition of the quantile, the total weighted mass strictly above this quantile is at most $\frac{\alpha_3}{2}$. Hence, for every $v\in\R^{T+1}$,
\[
\sum_{i\in S_U(v)}\bar w_i\leq \frac{\alpha_3}{2}.
\]

Let $K$ be a random index, independent of all data, with $P(K=i)=\bar w_i$ for $i\in[T+1]$. For a vector $v$, let $v^i$ denote the vector obtained by swapping the $i$-th and $(T+1)$-st entries of $v$. Since $w_i\leq 1$, we have $\bar w_i\leq \bar w_{T+1}$ for every $i\in[T]$. Therefore, replacing the unobserved $(T+1)$-st entry by $+\infty$ can only increase the corresponding upper weighted quantile. Thus, for every $v\in\R^{T+1}$,
\[
v_{T+1}
>
Q_{1-\frac{\alpha_3}{2}}
\left(
\sum_{j=1}^T \bar w_j\delta_{v_j}
+
\bar w_{T+1}\delta_{+\infty}
\right)
\quad\Longrightarrow\quad
K\in S_U(v^K).
\]
Applying this deterministic implication with $v=V^U = (\hat\Delta_1^U,\dots,\hat\Delta_T^U,\hat\Delta_{T+1}^U)$ gives
\[
P\left(\hat\Delta_{T+1}^U>\hat\Delta^U\right)
\leq
P\left(K\in S_U(V^{U,K})\right).
\]
Averaging over $K$, we have
\begin{align*}
P\left(K\in S_U(V^{U,K})\right)
&=
\sum_{i=1}^{T+1}
\bar w_i
P\left(i\in S_U(V^{U,i})\right) \leq
\sum_{i=1}^{T+1}
\bar w_i
P\left(i\in S_U(V^U)\right)
+
\sum_{i=1}^{T+1}
\bar w_i
d_{\mathrm{TV}}\left(V^U,V^{U,i}\right).
\end{align*}
Here, $V^{U,T+1}=V^U$, so the total-variation term for $i=T+1$ is zero. Therefore, by the definition of $\varepsilon_U$,
\[
\sum_{i=1}^{T+1}
\bar w_i
d_{\mathrm{TV}}\left(V^U,V^{U,i}\right)
=
\varepsilon_U.
\]
For the first term, the deterministic bound on the weighted mass of $S_U(V^U)$ gives
\[
\sum_{i=1}^{T+1}
\bar w_i
P\left(i\in S_U(V^U)\right)
=
\mathbb{E}\left[
\sum_{i\in S_U(V^U)}\bar w_i
\right]
\leq
\frac{\alpha_3}{2}.
\]
Combining the previous displays yields
\[
P\left(\hat\Delta_{T+1}^U\leq\hat\Delta^U\right)
\geq
1-\frac{\alpha_3}{2}-\varepsilon_U.
\]
Next, we prove the analogous lower-tail bound,
\[
P\left(\hat\Delta_{T+1}^L\geq \hat\Delta^L\right)
\geq
1-\frac{\alpha_3}{2}-\varepsilon_L.
\]
This follows by applying the upper-tail argument to the vector $-V^L$. Indeed, by the quantile convention for lower endpoints, $\hat\Delta_{T+1}^L<\hat\Delta^L$
is equivalent to
\[
-\hat\Delta_{T+1}^L
>
Q_{1-\frac{\alpha_3}{2}}
\left(
\sum_{j=1}^T \bar w_j\delta_{-\hat\Delta_j^L}
+
\bar w_{T+1}\delta_{+\infty}
\right).
\]
The same argument as above therefore gives
\[
P\left(\hat\Delta_{T+1}^L\geq \hat\Delta^L\right)
\geq
1-\frac{\alpha_3}{2}-\varepsilon_L,
\]
where we use the fact that total variation distance is unchanged by multiplying the endpoint vectors by $-1$.

Combining the lower- and upper-endpoint bounds by a union bound gives
\[
P\left(
[\hat\Delta_{T+1}^L,\hat\Delta_{T+1}^U]
\subseteq
[\hat\Delta^L,\hat\Delta^U]
\right)
\geq
1-\alpha_3-\varepsilon_L-\varepsilon_U.
\]
By Assumption~\ref{ass:classical_methods}, the counterfactual current-task gap interval satisfies $
P\left(
\Delta_{T+1}
\in
[\hat\Delta_{T+1}^L,\hat\Delta_{T+1}^U]
\right)
\geq
1-\alpha_2$.
Therefore, another union bound yields
\[
P\left(
\Delta_{T+1}
\in
[\hat\Delta^L,\hat\Delta^U]
\right)
\geq
1-\alpha_2-\alpha_3-\varepsilon_L-\varepsilon_U.
\]

Finally, by Assumption~\ref{ass:classical_methods}, the interval $[\tilde L,\tilde U]$ satisfies
$P(\tilde\theta\in[\tilde L,\tilde U])\geq 1-\alpha_1$.
On the event that $\tilde\theta\in[\tilde L,\tilde U]$ and $\Delta_{T+1}\in[\hat\Delta^L,\hat\Delta^U]$, we have
\[
\theta^*=\tilde\theta+\Delta_{T+1}
\in
[\tilde L+\hat\Delta^L,\tilde U+\hat\Delta^U]
=
\ci.
\]
Thus, by a final union bound,
\[
P(\theta^*\in\ci)
\geq
1-\alpha_1-\alpha_2-\alpha_3-\varepsilon_L-\varepsilon_U.
\]

\subsection{Proof of Theorem \ref{thm:multi_dimensional_targets}}

By the validity of the synthetic-data confidence region,
\[
P(\tilde\theta\in\widetilde{\mathrm{CI}})\geq 1-\alpha_1.
\]
Now introduce the counterfactual current-task gap region
\[
K_{T+1}=\Delta^{\theta^*,\alpha_2}(S,\tilde S),
\]
where $S=S_{T+1}$ is the unobserved real dataset from the current task. By Assumption~\ref{ass:classical_methods},
\[
P(\Delta_{T+1}\in K_{T+1})\geq 1-\alpha_2.
\]
As in the proof of Theorem~\ref{thm:main_thm}, task exchangeability and the conditional independence of the synthetic samples imply that
\[
(\cT_1,S_1,\tilde S_1),\dots,(\cT_{T+1},S_{T+1},\tilde S_{T+1})
\]
are exchangeable. Since each $K_j$ is a deterministic function of $(\cT_j,S_j,\tilde S_j)$, the regions $K_1,\dots,K_{T+1}$ are exchangeable. Because the nested family $\{\D_r\}$ is fixed, the scalar scores
\[
R_j=\inf\{r:K_j\subseteq\D_r\},
\qquad j=1,\dots,T+1,
\]
are also exchangeable.

Let $\hat R$ be the quantile in Algorithm~\ref{alg:multi_dimensional_targets}. The standard exchangeability argument gives
\[
P(R_{T+1}\leq \hat R)\geq 1-\alpha_3.
\]
On the event $R_{T+1}\leq \hat R$, nestedness implies
\[
K_{T+1}\subseteq\D_{\hat R}.
\]
Therefore, by a union bound,
\[
P(\Delta_{T+1}\in\D_{\hat R})
\geq
1-\alpha_2-\alpha_3.
\]
Finally, on the event that $\tilde\theta\in\widetilde{\mathrm{CI}}$ and $\Delta_{T+1}\in\D_{\hat R}$, we have
\[
\theta^*=\tilde\theta+\Delta_{T+1}
\in
\widetilde{\mathrm{CI}}\oplus\D_{\hat R}
=
\ci.
\]
A final union bound gives the claimed coverage guarantee.

\subsection{Proof of Theorem \ref{thm:finite_sample_target}}

Since $\theta^*(\tilde S)$ is observed, it suffices to study coverage of the current-task finite-sample gap
\[
\hat\Delta_{T+1}=\theta^*(S)-\theta^*(\tilde S)
\]
by the interval $[\hat\Delta^L,\hat\Delta^U]$.

Each gap $\hat \Delta_j=\theta_j(S_j)-\theta_j(\tilde S_j)$ is a deterministic function of $(\cT_j,S_j,\tilde S_j)$. By Assumption~\ref{ass:task_exchangeability}, the pairs $(\cT_1,S_1),\dots,(\cT_{T+1},S_{T+1})$ are exchangeable. Since each synthetic dataset $\tilde S_j$ is drawn from $\tilde\cP_j=\cG_{\cT_j}$ independently according to the same generator $\cG$ conditional on the task, the triples $(\cT_j,S_j,\tilde S_j)_{j=1}^{T+1}$ are exchangeable as well. Hence the scalars $\hat\Delta_1,\dots,\hat\Delta_{T+1}$ form an exchangeable sequence.

By the same uniform-rank property used in the proof of Theorem~\ref{thm:main_thm},
\[
P\left(\hat \Delta_{T+1}<\hat \Delta_{(k_L)}\right)\leq \frac{\alpha}{2}, \qquad
P\left(\hat \Delta_{T+1}>\hat \Delta_{(k_U)}\right)\leq \frac{\alpha}{2}.
\]
A union bound gives
\[
P\left(\hat \Delta_{T+1}\in[\hat\Delta^L,\hat\Delta^U]\right)\geq 1-\alpha.
\]
On this event, since $\theta^*(S)=\theta^*(\tilde S)+\hat\Delta_{T+1}$,
\[
\theta^*(S)\in
[\theta^*(\tilde S)+\hat\Delta^L,\ \theta^*(\tilde S)+\hat\Delta^U]
=\ci,
\]
which proves the coverage claim.

It remains to prove the sharper statement under the condition that there are no ties. If $\hat\Delta_1,\dots,\hat\Delta_{T+1}$ are almost surely distinct, then the rank of $\hat\Delta_{T+1}$ among the $T+1$ exchangeable gaps is exactly uniform. The event $\hat \Delta_{T+1}<\hat \Delta_{(k_L)}$
occurs exactly when this rank is at most $k_L$, and the event
$\hat \Delta_{T+1}>\hat \Delta_{(k_U)}$
occurs exactly when this rank is larger than $k_U$. 
Since the rank is uniform on $\{1,\dots,T+1\}$, we have
\[
P\left(\hat \Delta_{T+1}<\hat \Delta_{(k_L)}\right)
=
\frac{k_L}{T+1}, \qquad 
P\left(\hat \Delta_{T+1}>\hat \Delta_{(k_U)}\right)
=
\frac{T+1-k_U}{T+1}.
\]
By the definition $
k_L=\lfloor (T+1)\frac{\alpha}{2}\rfloor
$ and $
k_U=\lceil (T+1)\left(1-\frac{\alpha}{2}\right)\rceil$,
we have
\[
P\left(\hat \Delta_{T+1}\notin[\hat\Delta^L,\hat\Delta^U]\right)
=
\frac{k_L}{T+1}
+
\frac{T+1-k_U}{T+1}
=
\frac{2k_L}{T+1}.
\]
Equivalently,
\[
P\left(\theta^*(S)\in\ci\right)
=
P\left(\hat \Delta_{T+1}\in[\hat\Delta^L,\hat\Delta^U]\right)
=
1-\frac{2\left\lfloor (T+1)\frac \alpha 2\right\rfloor}{T+1}.
\]
Finally, since $
\lfloor (T+1)\frac{\alpha}{2}\rfloor
\leq
(T+1)\frac{\alpha}{2},$
we recover
\[
P\left(\theta^*(S)\in\ci\right)\geq 1-\alpha.
\]
Moreover, since
$\lfloor (T+1)\frac{\alpha}{2}\rfloor
>
(T+1)\frac{\alpha}{2}-1$,
we also have
\[
P\left(\theta^*(S)\in\ci\right)
<
1-\alpha+\frac{2}{T+1}.
\]

\section{A weaker task exchangeability condition}
\label{sec:weaker_condition}

Our main result assumes exchangeability of the task--dataset pairs $(\cT_1,S_1),\dots,(\cT_{T+1},S_{T+1})$,
where $\cT_{T+1}=\cT^*$ and $S_{T+1}=S$ denotes the counterfactual real dataset from the target task. This assumption is convenient because it implies exchangeability of the estimated gap intervals used by Algorithm~\ref{alg:main_algo}. A weaker condition would only assume exchangeability at the level of the underlying tasks, not datasets. For example, the historical sample sizes $n_j$ may be chosen systematically, or may vary for reasons unrelated to the task-generating process.

We therefore consider the following weaker condition.

\begin{assumption}[Weaker task exchangeability]
\label{ass:task_exchangeability_weaker}
The tasks $\cT_1,\dots,\cT_{T+1}$, where $\cT_{T+1}=\cT^*$, are exchangeable.
\end{assumption}

Under Assumption~\ref{ass:task_exchangeability_weaker}, the confidence intervals computed from the historical samples need not be exchangeable, because they depend on the datasets $S_j$ and their sample sizes. However, the latent real--synthetic gaps
\[
\Delta_j = \theta_j(\cP_j)-\theta_j(\tilde \cP_j),
\qquad j=1,\dots,T+1,
\]
\emph{are} exchangeable, since each $\Delta_j$ is a deterministic function of the task $\cT_j$ and the synthetic distribution $\tilde \cP_j=\cG_{\cT_j}$ induced by the generator. We can therefore calibrate inference on the target task using the order statistics of $\Delta_1,\dots,\Delta_T$. Since these gaps are not observed, we first construct simultaneous confidence intervals for the historical gaps $\Delta_j$ and then use them to conservatively bound the relevant order statistics.

This leads to Algorithm~\ref{alg:weaker_algo}. Compared with Algorithm~\ref{alg:main_algo}, the procedure is more conservative: it uses a Bonferroni correction across the $T$ historical gap intervals, rather than relying on exchangeability of the estimated intervals directly.

\begin{algorithm}[t]
\caption{Inference under weaker task exchangeability}
\label{alg:weaker_algo}
\begin{algorithmic}[1]
\Require current task $\cT^*$, historical tasks with real data $(\cT_1,S_1),\dots,(\cT_T,S_T)$, synthetic data generator $\cG$, synthetic sample size $N$, error levels $\alpha_1,\alpha_2,\alpha_3\in(0,1)$
\State Draw $N$ synthetic samples for the current task: $\tilde S\stackrel{\mathrm{i.i.d.}}{\sim}\tilde \cP$, where $\tilde \cP=\cG_{\cT^*}$
\State Draw $N$ synthetic samples for each historical task: $\tilde S_j\stackrel{\mathrm{i.i.d.}}{\sim}\tilde \cP_j$, where $\tilde \cP_j=\cG_{\cT_j}$, for all $j\in[T]$
\State Compute a confidence interval for the current synthetic-data target: $
[\tilde L,\tilde U]
=
\mathrm{CI}^{\theta^*,\alpha_1}(\tilde S)$
\State For each historical task $j\in[T]$, compute a confidence interval for the gap:
$[\hat\Delta_j^L,\hat\Delta_j^U]
=
\Delta^{\theta_j,\alpha_2/T}(S_j,\tilde S_j)$
\State Let $
\hat\Delta_L=\hat\Delta_{(k_L)}^L$ and $\hat\Delta_U=\hat\Delta_{(k_U)}^U$, where $k_L=\lfloor (T+1)\frac{\alpha_3}{2}\rfloor$ and $k_U=\lceil (T+1)\left(1-\frac{\alpha_3}{2}\right)\rceil$ 
\Ensure Confidence interval for $\theta^*$: $\mathrm{CI} = [\tilde L+\hat\Delta_L,
\tilde U+\hat\Delta_U]$
\end{algorithmic}
\end{algorithm}

\begin{theorem}
\label{thm:weaker}
Suppose Assumptions~\ref{ass:classical_methods} and~\ref{ass:task_exchangeability_weaker} hold. Then, the confidence interval output by Algorithm~\ref{alg:weaker_algo} satisfies
\[
P(\theta^*\in \mathrm{CI})
\ge
1-(\alpha_1+\alpha_2+\alpha_3).
\]
\end{theorem}

\begin{proof}
By Assumption~\ref{ass:classical_methods}, the synthetic-data interval satisfies
\[
P(\tilde\theta\in[\tilde L,\tilde U])
\ge
1-\alpha_1.
\]
It remains to show that $[\hat\Delta_L,\hat\Delta_U]$ covers the current gap $
\Delta_{T+1}
=
\theta^*-\tilde\theta
=
\theta_{T+1}(\cP_{T+1})
-
\theta_{T+1}(\tilde\cP_{T+1})$
with probability at least $1-\alpha_2-\alpha_3$.

Since each $\Delta_j$ is a deterministic function of the task $\cT_j$, Assumption~\ref{ass:task_exchangeability_weaker} implies that
\[
\Delta_1,\dots,\Delta_{T+1}
\]
are exchangeable. Let $\Delta_{(1)}\le \cdots \le \Delta_{(T)}$
denote the order statistics of the historical gaps $\Delta_1,\dots,\Delta_T$, with the conventions $\Delta_{(0)}=-\infty$ and $\Delta_{(T+1)}=\infty$. Exchangeability implies that $\Delta_{T+1}$ has a uniform rank among all the gaps, and thus
\[
P\left(
\Delta_{T+1}\in
[\Delta_{(k_L)},\Delta_{(k_U)}]
\right)
\ge
1-\alpha_3,
\]
where $k_L=\lfloor (T+1)\frac{\alpha_3}{2}\rfloor$ and $k_U=\lceil (T+1)\left(1-\frac{\alpha_3}{2}\right)\rceil$.

Now define the simultaneous historical coverage event
\[
E
=
\bigcap_{j=1}^T
\left\{
\Delta_j\in
[\hat\Delta_j^L,\hat\Delta_j^U]
\right\}.
\]
By Assumption~\ref{ass:classical_methods}, for each $j\in[T]$, $
P\left(
\Delta_j\in
[\hat\Delta_j^L,\hat\Delta_j^U]
\right)
\ge
1-\alpha_2/T$.
Therefore, by a union bound,
\[
P(E)\ge 1-\alpha_2.
\]

On the event $E$, we have $
\hat\Delta_j^L\le \Delta_j\le \hat\Delta_j^U$ for all $j\in[T]$.
Hence the corresponding order statistics satisfy
\[
\hat\Delta_{(k_L)}^L
\le
\Delta_{(k_L)}
\text{ and }
\hat\Delta_{(k_U)}^U
\ge
\Delta_{(k_U)}.
\]
Equivalently, on $E$, $
[\Delta_{(k_L)},\Delta_{(k_U)}]
\subseteq
[\hat\Delta_L,\hat\Delta_U].
$
Combining this with the fact that $[\Delta_{(k_L)},\Delta_{(k_U)}]$ covers $\Delta_{T+1}$ with probability $1-\alpha_3$ gives
\[
P\left(
\Delta_{T+1}\in
[\hat\Delta_L,\hat\Delta_U]
\right)
\geq
P\left(
\Delta_{T+1}\in
[\Delta_{(k_L)},\Delta_{(k_U)}],
E
\right)  \
\geq
1-\alpha_2-\alpha_3.
\]

Finally, on the event that $
\tilde\theta\in[\tilde L,\tilde U]$ and $\Delta_{T+1}\in[\hat\Delta_L,\hat\Delta_U]$,
we have
\[
\theta^*
=
\tilde\theta+\Delta_{T+1}
\in
[\tilde L+\hat\Delta_L,
\tilde U+\hat\Delta_U]
=
\mathrm{CI}.
\]
A final union bound yields
\[
P(\theta^*\in\mathrm{CI})
\ge
1-(\alpha_1+\alpha_2+\alpha_3).
\]
\end{proof}

\end{document}